\documentclass[a4paper,12pt]{article}

\usepackage{amsthm}
\usepackage{amsmath}
\usepackage{amssymb}
\usepackage{latexsym}
\usepackage{float}
\usepackage{diagbox} 
\usepackage{threeparttable}
\usepackage[textwidth=18cm,textheight=20cm]{geometry}
\usepackage[square,sort,comma,numbers]{natbib}

\newtheorem{theorem}{Theorem}[section]
\newtheorem{proposition}[theorem]{Proposition}
\newtheorem{lemma}[theorem]{Lemma}
\newtheorem{corollary}[theorem]{Corollary}
\newtheorem{remark}[theorem]{Remark}
\newtheorem{definition}[theorem]{Definition}
\newtheorem{example}[theorem]{Example}
\newtheorem{assumption}[theorem]{Assumption}
\newcommand{\ud}{\mathrm{d}}
\newcommand{\eqdefr}{=\mathrel{\mathop:}}
\newcommand{\eqdefl}{\mathrel{\mathop:}=}
\newcommand{\Q}{\mathbb{Q}}
\newcommand{\FF}{\mathcal{F}}
\newcommand{\f}{\mathbb{F}}
\newcommand{\A}{\mathcal{A}}
\newcommand{\B}{\mathcal{B}}
\newcommand{\M}{\mathcal{M}}
\newcommand{\K}{\mathcal{K}}
\newcommand{\C}{\mathcal{C}}
\newcommand{\s}{\hat{s}}
\newcommand{\tha}{\hat{\theta}}
\newcommand{\LL}{\mathcal{L}}
\newcommand{\R}{\mathbb{R}}
\newcommand{\1}{\textbf{1}}
\newcommand{\N}{\mathbb{N}}
\newcommand{\PP}{\mathbb{P}}
\newcommand{\Cov}{\mathbf{Cov}}
\newcommand{\Var}{\mathbf{Var}}
\newcommand{\Corr}{\mathbf{Corr}}

\newcommand{\ackname}{Acknowledgements}

\newcommand{\e}{\varepsilon}
\DeclareMathOperator{\inter}{int}
\DeclareMathOperator{\conv}{conv}
\DeclareMathOperator{\supp}{supp}
\DeclareMathOperator{\relint}{ri}
\DeclareMathOperator{\interior}{int}

\usepackage{graphicx}

\newcommand{\F}{\mathscr{F}}

\author{\normalsize{$\text{Mikl\'os R\'asonyi}^1$}
\normalsize{$\text{Hasanjan Sayit}^2$}  \\
\footnotesize{$^1\text{Rényi Institute, Budapest, Hungary}$}\\
\footnotesize{$^2\text{Xi'Jiao Liverpool University, Suzhou, China}$}}

\date{September 20, 2021}

\begin{document}
\title{Exponential utility maximization in small/large financial markets}
\date{\today}

\maketitle



\abstract{This note studies optimal portfolio selection problems under the exponential utility function $U(x)=-e^{-ax}, a>0,$ in a  market with multiple risky assets with return vector given by normal mean-varaince mixture (NMVM) models (models given by (\ref{one}) in the introduction). An agent with initial capital $W_0$ invests her wealth on a risk-free asset with return $r_f$ and  multiple risky assets with NMVM return vector. In this note, under the assumption of linear independence of $\gamma$ and $\mu-\1r_f$,  it is shown that the agents optimal portfolio is given explicitly by
\[
x^{\star}=\frac{1}{aW_0}[\Sigma^{-1}\gamma +\frac{1}{EZ}\Sigma^{-1}(\mu-\1r_f)].
\]
This note also studies large financial markets based on normal mean-variance mixture models. The main result in this part of the note shows  that the optimal exponential utilities based on small markets converge to the optimal exponential utility in the large financial market. This shows, in particular, that to reach the best utility level investors need to diversify their investment to include infinitely many assets into their portfolio and with portfolios based on only finitely many assets they never be able to reach the optimum level of utility.
}

\vspace{0.1in}
\vspace{0.1in}

\vspace{0.1in}

\section{Introduction}
We consider a frictionless financial market with $d+1$ assets. We assume the first asset is a risk-free asset with risk-free interest rate $r_f$ and the remaining $d$ assets are risky assets with returns modelled by an $d-$dimensional random vector $X$. In this note, we assume that $X$ follows normal mean-variance mixture (NMVM) distribution as follows
\begin{equation}\label{one}
X\overset{d}{=}\mu+\gamma Z+\sqrt{Z}AN,
\end{equation}
where $\mu \in \R^d$ is location parameter, $\gamma \in \R^d$ controls the skewness, $Z\sim G$ is a non-negative random variable with distribution function $G$,  $A\in \R^{d\times d}$ is a symmetric and positive definite  $d\times d$ matrix of real numbers, $N\sim N(0, I)$ is a $d-$dimensional  Gaussian random vector with identity co-variance matrix $I$ in $\R^d\times \R^d$,  and $N$ is independent from the mixing distribution $Z$. 

In this paper we use the following notations. For any vectors $x=(x_1, x_2, \cdots, x_d)^T$ and $y=(y_1, y_2, \cdots, y_d)^T$ in $\R^d$, where the superscript $T$ stands for the transpose of a vector, $<x, y>=x^Ty=\sum_{i=1}^dx_iy_i$ denotes the scalar product of the vectors $x$ and $y$, and $|x|=\sqrt{\sum_{i=1}^dx_i^2}$ denotes the Euclidean norm of the vector $x$.  We sometimes use the short hand notation $X\sim N(\mu+\gamma z, z\Sigma ) \circ G$ for (\ref{one}). $\R$ denotes the set of real numbers and $\R_+=[0, +\infty)$ denotes the set of non-negative real-numbers. Following the same notations of \cite{Mark-Yor-GGC}, $\mathcal{J}$ denotes the family of infinitely divisible random variables on $\R_+$, $\mathcal{S}$ denotes the set of self-decomposable random variables on $\R_+$, and $\mathcal{G}$ denotes the class of generalized gamma convolutions (GGCs) on $\R_+$ that will be introduced later. The Laplace transformation of any distribution $G$ is denoted by $\mathcal{L}_G(s)=\int e^{-sy}G(dy)$. A gamma random variable with density function $f(x)=\frac{1}{\Gamma(\alpha)\beta^{\alpha}}x^{\alpha-1}e^{-x/\beta}$
is denoted by $G=G(\alpha, \beta)$.

A prominent example of the NMVM models is generalized hyperbolic (GH) distributions, where the mixing distribution $Z$ follows generalized inverse Gaussian (GIG) distribution denoted as $GIG(\lambda, a, b)$. The probability density function of a GIG distribution, denoted by $f_{GIG}(\lambda, a, b)$, takes the following form
\begin{equation}\label{gig}
 f_{GIG}(x; \lambda, a, b)= (\frac{b}{a})^{\lambda}\frac{1}{K_{\lambda}(ab)}x^{\lambda-1}e^{-\frac{1}{2}(a^2x^{-1}+b^2x)}1_{(0, +\infty)}(x),  
\end{equation}
where $K_{\lambda}(x)$ denotes the modified Bessel function of third kind with index $\lambda$ and the allowed parameter ranges for $\lambda, a, b$ in (\ref{gig}) are (i) $a\geq 0, b>0$ if $\lambda>0$, (ii)
$a>0, b\geq 0$ if $\lambda<0$, (iii) $a>0, b>0$ if $\lambda=0$. Here the case $a=0$ in (i) or the case $b=0$ in (ii)  above need to be understood in limiting cases of (\ref{gig}) and in these special cases we have
\begin{equation}\label{gainga}
\begin{split}
 f_{GIG}(x; \lambda, 0, b)&=(\frac{b^2}{2})^{\lambda}\frac{x^{\lambda-1}}{\Gamma (\lambda)}e^{-\frac{b^2}{2}x}1_{(0, +\infty)}(x), \; \; \; \lambda>0, \\
  f_{GIG}(x; \lambda, a, 0)&=(\frac{2}{a^2})^{\lambda}\frac{x^{\lambda-1}}{\Gamma (-\lambda)}e^{-\frac{a^2}{2x}}1_{(0, +\infty)}(x), \; \; \; \lambda<0,
 \end{split}
\end{equation}
where $\Gamma(x)$ denotes the Gamma function.
Here $f_{GIG}(x; \lambda, 0, b)$ is the density function of a Gamma distribution $G(\lambda, \frac{2}{b^2})$ and $f_{GIG}(x; \lambda, a, 0)$ is the density function of a inverse Gamma distribution $iG(\lambda, \frac{a^2}{2} )$.

The GH distribution in dimension $d$ is denoted by $GH_d(\lambda, \alpha, \beta, \delta, \mu, \Sigma)$ and it satisfies $GH_d(\lambda, \alpha, \beta, \delta, \mu, \Sigma)\sim N(\mu+z\Sigma \beta, z\Sigma)\circ GIG(\lambda, \delta, \sqrt{\alpha^2-\beta^T\Sigma \beta})$. The parameter ranges of this distribution is $\lambda \in \R, \; \alpha, \delta \in \R_+,\; \beta, \mu \in \R^d$ and (i') $\delta\geq 0,\;  0\le \sqrt{\beta^T\Sigma \beta}<\alpha$ if $\lambda>0$, (ii') $\delta> 0,\;  0\le \sqrt{\beta^T\Sigma \beta}<\alpha$ if $\lambda=0$, (iii') $\delta> 0,\;  0\le \sqrt{\beta^T\Sigma \beta}\le \alpha$ if $\lambda<0$.
The class of GH distributions include two popular models in finance: if $\lambda =-\frac{1}{2}$ we have normal inverse Gaussian distribution which is denoted by $NIG_d(\alpha, \beta, \delta, \mu, \Sigma)$ and when $\lambda=\frac{1+d}{2}$ we have the class of hyperbolic distributions denoted by $HYP_d(\alpha, \beta, \delta, \mu, \Sigma)$. As in the case of the GIG distributions, the case 
$\delta=0$ in (i') above and the case 
$\sqrt{\beta^T\Sigma \beta}=\alpha$ or $\alpha=0$ in (iii') above needs to be understood as limiting cases of the GH distributions. If $\lambda>0, \delta \rightarrow 0$ in case (i') above then
\begin{equation}\label{4}
 GH_d(\lambda, \alpha, \beta, \delta, \mu, \Sigma)\overset{w}{\rightarrow } 
 N_d(\mu+z\Sigma \beta, z\Sigma)\circ G(\lambda, \frac{\alpha^2-\beta^T \Sigma \beta}{2})=:VG_d(\lambda, \alpha, \beta,\mu, \Sigma), 
\end{equation}
where $\overset{w}{=}$ denotes weak convergence of distributions and $VG_d$
represents the class of variance gamma distributions. If $\lambda<0$ and $\alpha \rightarrow 0$ as well as $\beta \rightarrow 0$ in case (iii') above we have the shifted $t$ distributions with degrees of freedom $-2\lambda$
\begin{equation}\label{5}
 GH_d(\lambda, \alpha, \beta, \delta, \mu, \Sigma)\overset{w}{\rightarrow } 
 N(\mu, z\Sigma)\circ iG(\lambda, \frac{\delta^2}{2})=:t_d(\lambda, \delta, \mu, \Sigma). 
\end{equation}
If $\alpha \rightarrow \infty, \delta \rightarrow \infty$ and  $\frac{\delta}{\alpha}\rightarrow \sigma^2<\infty$, we have the following that shows that the Normal random vectors are limiting cases of the GH distributions
\begin{equation}\label{6d}
 GH_d(\lambda, \alpha, \beta, \delta, \mu, \Sigma)\overset{w}{\rightarrow } 
 N(\mu+z\Sigma \beta, z\Sigma)\circ \epsilon_{\sigma^2}=: N(\mu+\sigma^2\Sigma \beta, \sigma^2\Sigma), 
\end{equation}
where $\epsilon_{\sigma^2}$ is the dirac function that equals to $1$ when $z=1$ and equals to zero otherwise. All these normal inverse Gaussian, hyperbolic, variance gamma, and student $t$ distributions are very popular models in finance, see  \cite{Kercheval}, \cite{Bingham_NICHOLAS_H_And_Kiesel_Rudigger_2001}, \cite{Birge-Bedoya}, \cite{Eberlein_Ernst_And_Keller_Ulrich_1995}, \cite{Hellmich_And_Kassberger_2011}, \cite{Madan_Dilip_B_And_Carr_Peter_P_And_Chang_Eric_C_1998}, \cite{Madan_Dilip_B_And_Seneta_Eugene_1990}, \cite{Prause_Karsten_1999}, \cite{Menca2005EstimationAT} for this.

The class of GIG distributions belong to the class of GGCs. A positive random variable $Z$ is a GGC, without translation term, if there exists a positive Radon measure $\nu$ on $\R_+$ such that
\begin{equation}\label{GGClap}
 \mathcal{L}_Z(s)=Ee^{-sZ}=e^{-\int_0^{\infty}\ln (1+\frac{s}{z})\nu(dz)},   
\end{equation}
with 
\begin{equation}\label{8}
 \int_0^1|lnx| \nu(dx)<\infty, \; \; \int_1^{\infty}\frac{1}{x}\nu(dx)<\infty.   
\end{equation}
The measure $\nu$ is called Thorin's measure associated with $Z$.
For the definition of the GGCs see the survey paper \cite{Mark-Yor-GGC}. In Proposition 1.1 of \cite{Mark-Yor-GGC}, it was shown that
any GGC random variable can be written as Wiener-Gamma integral
\begin{equation}\label{9}
Z=\int_0^{\infty}h(s)d\gamma_s,    
\end{equation}
where $h(s): \R_+ \rightarrow \R_+$ is a deterministic function with $\int_0^{\infty}ln(1+h(s))ds<\infty$ and $\{\gamma_s\}$ is a standard Gamma process with L\'evy measure $e^{-x}\frac{dx}{x}, x>0$. 

Proposition 1.23 of \cite{Hammerstein_EAv_2010} shows that the class of GIG random variables belongs to the class GGC. It provides the description of the corresponding Thorin's measures (in terms of the functions $U_{GIG}$ in the Proposition) for all the cases of parameters of GIG. The class of GGC distributions are rich as stated in the introduction of \cite{Mark-Yor-GGC} and we have the relation 
$\mathcal{G}\subset \mathcal{S}\subset \mathcal{J}$. In our model (\ref{one}) the mixing distribution $Z$ can be any distribution in $\mathcal{J}$. In fact, $Z$ can  be any non-negative random variable.

Given an initial endowment $W_0>0$, the investor must determine the portfolio weights x on the $n$ risky assets such that the expected utility of the next period wealth is maximized. The wealth that corresponds to portfolio weight $x$ on the risky assets is given by
\begin{equation}\label{wealth}
\begin{split}
W(x)=&W_0[1+(1-x^T1)r_f+x^TX] \\
=&W_0(1+r_f)+W_0[x^T(X-\1 r_f)]
\end{split}
\end{equation}
 and the investor's problem is
\begin{equation}\label{L2}
\max_{x\in D}\; EU(W(x)),\\
\end{equation}
for some domain $D$ of the portfolio set $D$. Note here that $x$ represents the portfolio weights on the risky assets and $1-x^T\1$ is the proportion of the initial wealth invested on the risk free asset. The portfolio weights $x$ on risky assets are allowed to be any vector in $D$.

The main goal of this paper is to discuss the solution of the problem (\ref{L2}) for exponential utility function $U$ when the returns of the risky assets have NMVM distribution as in  (\ref{one}). These type of utility maximization problems in one period models were studied in many papers in the past, see \cite{Madan_Mcphail}, \cite{Madan_Yen}, \cite{Pirvu_Kwak}, \cite{Zakamouline_Koekabakkar}, \cite{Birge_Chavez_2016}. Especially, the recent paper \cite{Birge-Bedoya} made an interesting observation that, with generalized hyperbolic models and with exponential utility, the optimal portfolios of the corresponding expected utility maximization problems can be written as a sum of two portfolios that are determined by the location and skewness parameters of the model (\ref{one}) separately. The present paper, extends their result to more general class of NMVM models as a compliment. 

The paper is organized as follows. In section 2 below we present closed form solution for optimal portfolio when the utility function $U$ is exponential. In section 3, we show that the optimal expected utilities in small financial markets converge to an overall best expected utility in a large financial market. In section 4  we present examples as applications of our results.

\section{Closed form solution for optimal portfolios under exponential utility}

In this section, we study the solution of the problem (\ref{L2}) when the utility function of the investor is exponential 
\begin{equation}
U(W)=-e^{-aW}, a>0,
\end{equation}
and when the investment opportunity set consists of the above stated $d+1$  assets. Below we obtain an expression that relates $EU(W)$ to the Laplace transformation of the mixing distribution $Z$ as in (\ref{L3}) below. First observe that we have
 \begin{equation}\label{13-new}
 W(x)\overset{d}{=}W_0(1+r_f)+W_0[x^T(\mu-\1r_f)+x^T\gamma Z+\sqrt{x^T\Sigma x}\sqrt{Z}N(0, 1)].    
 \end{equation}
 
\begin{lemma}\label{2.1} For any portfolio $x\in \R^d$ such that $EU(W(x))$ is finite we have
\begin{equation}\label{L3}
EU(W(x))=-e^{-aW_0(1+r_f)}e^{-aW_0x^T(\mu-\1 r_f)}\mathcal{L}_Z\Big(aW_0x^T\gamma - \frac{a^2W_0^2}{2}x^T\Sigma x\Big ),    
\end{equation}
where $\mathcal{L}_Z(s)=Ee^{-sZ}$ is the Laplace transformation of $Z$.
\end{lemma}
\begin{proof} From (\ref{13-new}), we have
\begin{equation*}
\begin{split}
EU(W(x))=&-Ee^{-aW_0(1+r_f)-aW_0\Big [x^T(\mu-\1 r_f)+x^T\gamma Z+\sqrt{x^T\Sigma x}\sqrt{Z}N(0, 1)\Big ]}\\
=&-e^{-aW_0(1+r_f)}e^{-aW_0x^T(\mu-\1 r_f)}\int_0^{+\infty}Ee^{-aW_0x^T\gamma z-aW_0\sqrt{x^T\Sigma x}\sqrt{z}N(0, 1)}f_Z(z)dz\\
=&-e^{-aW_0(1+r_f)}e^{-aW_0x^T(\mu-\1 r_f)}\int_0^{+\infty}e^{-aW_0x^T\gamma z}Ee^{-aW_0\sqrt{x^T\Sigma x}\sqrt{z}N(0, 1)}f_Z(z)dz\\
=&-e^{-aW_0(1+r_f)}e^{-aW_0x^T(\mu-\1 r_f)}\int_0^{+\infty}e^{-aW_0x^T\gamma z}e^{\frac{a^2W_0^2}{2}x^T\Sigma xz}f_Z(z)dz\\
=&-e^{-aW_0(1+r_f)}e^{-aW_0x^T(\mu-\1 r_f)}\int_0^{+\infty}e^{-(aW_0x^T\gamma-\frac{a^2W_0^2}{2}x^T\Sigma x)z}f_Z(z)dz\\
=&-e^{-aW_0(1+r_f)}e^{-aW_0x^T(\mu-\1 r_f)}\mathcal{L}_Z(aW_0x^T\gamma-\frac{a^2W_0^2}{2}x^T\Sigma x).\\
\end{split}
\end{equation*}
\end{proof}
\begin{remark} \label{rem-mu-r}If $\mu-\1r_f=0$ in our model (\ref{one}), from (\ref{L3}) we have 
\begin{equation*}
EU(W(x))=-e^{-aW_0(1+r_f)}\mathcal{L}_Z\Big(aW_0x^T\gamma - \frac{a^2W_0^2}{2}x^T\Sigma x\Big ).    
\end{equation*}
Since $\mathcal{L}_Z(s)$ is a strictly decreasing function, the expected utility maximization problem becomes the maximization problem of the quadratic function  $aW_0x^T\gamma - \frac{a^2W_0^2}{2}x^T\Sigma x$ in this case.
Therefore for the rest of the paper, we assume that 
our model (\ref{one}) is such that $\mu-\1 r_f\neq 0$. Also we assume $Z\neq 0$ with positive probability.
\end{remark}
\begin{remark} By using the relation (\ref{L2}) and by checking the first order condition for optimality it is easy to see that the optimal portfolio $x^{\star}$ satisfies the following relation
\begin{equation}\label{optxstar}
x^{\star}=\frac{1}{aW_0}[\Sigma^{-1}\gamma-\frac{\mathcal{L}(g(x^{\star}))}{\mathcal{L}'(g(x^{\star}))}\Sigma^{-1}(\mu-\1r_f)],
\end{equation}
where $g(x)$ is given in the expression (\ref{Gx}) below. There are several questions that one needs to address when applying the direct approach (\ref{optxstar}) in obtaining the optimal portfolio $x^{\star}$: (i) if the function $x\rightarrow EU(W(x))$ is continuously differentiable (ii) if the optimal portfolio is the interior point of the corresponding domain (iii) if the equation (\ref{optxstar}) has unique solution.  After these questions are addressed the next challenge becomes how to compute $x^{\star}$ numerically. This problem is not trivial if the dimension $d$ is a large number, i.e., $x\in \R^d$ for large $d$. To overcome these problems, in this paper we take different approach and obtain $x^{\star}$ in near closed form: to calculate $x^{\star}$ we only need to find the minimizing point of a convex function on the real line. 
\end{remark}

The above Lemma \ref{L3}, expresses the expected utility in terms of a linear function $x^T(\mu-\1r_f)$  and  a quadratic function $aW_0x^T\gamma - \frac{a^2W_0^2}{2}x^T\Sigma x$ of the portfolio $x\in \R^n$. For convenience, we introduce the following notations 
\begin{equation}\label{Gx}
\begin{split}
g(x)=:&aW_0x^T\gamma-\frac{a^2W_0^2}{2}x^T\Sigma x,\\
G(x)=:&e^{-aW_0x^T(\mu-\1 r_f)}\mathcal{L}_Z\Big (aW_0x^T\gamma -\frac{a^2W_0^2}{2}x^T\Sigma x\Big ),\\
=&e^{-aW_0x^T(\mu-\1 r_f)}\mathcal{L}_Z\Big (g(x)\Big ).
\end{split}
\end{equation}

Then the relation (\ref{L3}) becomes
\begin{equation}
EU(W)=-e^{aW_0(1+r_f)}G(x)=-e^{aW_0(1+r_f)}e^{-aW_0x^T(\mu-\1 r_f)}\mathcal{L}_Z\big (g(x)\big ).   
\end{equation}
Therefore we have the following obvious relation
\begin{equation}\label{arG}
\arg\max_{x\in D}EU(W)=\arg\min_{x\in D}G(x)
\end{equation}
for any domain $D\in \R^d$ of the portfolio set. Note here that the equality in (\ref{arG}) means the equality of two sets if the optimizing points are more than one.

Our goal in this section is to give closed form solution for the problem (\ref{L2}) for some  domains of the portfolio set. Before we start our analysis, we first present the following example.

\begin{example}\label{ex2.2} Consider the model (\ref{one}) with $\gamma=0$ and with the mixing distribution $Z\sim e^{N(0, 1)}$. Then for any $x\neq 0$ we have
\[
EU(W(x))=-\infty.
\]
To see this, assume that there is a $x\neq 0$ such that $EU(W(x))$ is finite. Then by Lemma  \ref{2.1} we have
\[
EU(W(x))=-e^{-aW_0(1+r_f)}e^{-aW_0x^T(\mu-\1 r_f)}\mathcal{L}_Z\Big(- \frac{a^2W_0^2}{2}x^T\Sigma x\Big ).  
\]
For any $x\neq 0$ we have $x^T\Sigma x>0$ as $\Sigma$ is positive definite by the assumption of the model (\ref{one}). Now it is well known that when $Z\sim e^{N(0, 1)}$ we have $\mathcal{L}_Z(s)=+\infty$ whenever $s<0$. Therefore $\mathcal{L}_Z\Big(- \frac{a^2W_0^2}{2}x^T\Sigma x\Big )=+\infty$ whenever $x\neq 0$ and this contradicts with the finiteness assumption of  $EU(W(x))$ made above. Thus we have $EU(W(x))=-\infty$ whenever $x\neq 0$. Therefore the problem (\ref{L2}) does not have a solution when the domain $D$ does not include the zero vector in it. But if $0\in D$, then $x=0$ is the optimal portfolio and $\max_{x\in D}\; EU(W(x))=-e^{-aW_0(1+r_f)}$. This case corresponds to 
investing all the initial wealth $W_0$ on the risk-free asset as an optimal portfolio. 
\end{example}

The above Example \ref{ex2.2} shows that when the model (\ref{one}) satisfies the conditions in the example and when $0\in D$, the zero portfolio $x=0$ is an optimal portfolio as when $x\neq 0$ one has $EU(W(x))=-\infty$ always. It is obvious that, in this case, the function  $x\rightarrow EU(W(x))$ is not differentiable at $x=0$. Therefore we call $x=0$ irregular solution for the optimization problem (\ref{arG}). Before we give formal definition of irregularity, 
we first introduce the following definition.

\begin{definition}\label{defs0} For any mixing distribution $Z$, if $\mathcal{L}_Z(s)<\infty$ for all $s\in \R$ we set $\s=-\infty$ and if $\mathcal{L}_Z(s)<\infty$ for some $s\in \R$ and $\mathcal{L}_Z(s)=+\infty$ for some $s\in \R$, we let $\s$ be the real number such that 
\begin{equation}\label{s00}
\mathcal{L}_Z(s)=Ee^{-sZ}<\infty, \; \forall s>\s\;\;  \mbox{and} \;\; \mathcal{L}_Z(s)=Ee^{-sZ}=+\infty, \; \forall s<\s. 
\end{equation}
We call $\s$ the critical value (CV) (we use the acronym CV-L from now on, where $L$ implies that it is CV in the context of  Laplace transformation. One can also define this CV in the context of moment generating functions and in this case an acronym CV-M can be used) of $Z$ under Laplace transformation in this paper. Observe that since $Z$ is non-negative random variable we always have $\s\le 0$.
\end{definition}

\begin{remark}\label{rem2.2} In the above definition \ref{defs0}, the value of $\mathcal{L}_Z(s)$ at $s=\s$ is not specified. Both of the cases 
$\mathcal{L}_Z(\s)<\infty$ and $\mathcal{L}_Z(\s)=+\infty$ are possible.
For example if $Z\sim e^{N(0, 1)}$, then $\s=0$ and clearly $\mathcal{L}_Z(0)=1<\infty$. If $Z\sim x^{\alpha-1}e^{-x/\beta}/ [\Gamma(\alpha)\beta^{\alpha}]$ is a Gamma distribution, then $\mathcal{L}_Z(s)=1/[(1+\beta s)^{\alpha}]$. In this case $\s=-1/\beta$ and we have $\mathcal{L}_Z(\s)=+\infty$. 
\end{remark}

Below we define some domains for the portfolio set.

\begin{equation}\label{Sa}
\begin{split}
S_a=:&\{x\in \R^d: aW_0x^T\gamma-\frac{a^2W_0^2}{2}x^T\Sigma x>\s\},\\
\partial S_a=:&\{x\in \R^d: aW_0x^T\gamma-\frac{a^2W_0^2}{2}x^T\Sigma x= \s\},\\
\bar{S}_a=:&S_a\cup \partial S_a.
\end{split}
\end{equation}

\begin{remark}\label{rem2.55} Our main objective in this section is to find closed form solution for the optimal portfolio for the problem
\begin{equation}\label{L22}
\max_{x\in \R^d}\; EU(W(x)).    
\end{equation}
 The following relations are easy to see \begin{equation}\label{saprob}
\max_{x\in \R^d}\; EU(W(x))=\max_{x\in S_a}EU(W(x)),    
\end{equation}
if $\mathcal{L}_Z(\s)=+\infty$ and 
\begin{equation}\label{barsaprob}
\max_{x\in \R^d}\; EU(W(x))=\max_{x\in \bar{S}_a}EU(W(x)),    
\end{equation}
if $\mathcal{L}_Z(\s)<+\infty$.
Observe here that if $\s<0$, then $S_a$ is a nonempty set as the zero vector $x=0$ is in it. If $\s=0$, then the set $\bar{S}_a$ is nonempty as $x=0$ is in it.
\end{remark}

In this section we attempt  to give closed form solutions for the problems (\ref{saprob}) and (\ref{barsaprob}) above. Our approach for this is based on the following idea: we fix the term $x^T(\mu-\1r_f)$ at some constant level $c$ and optimize the quadratic term $aW_0x^T\gamma - \frac{a^2W_0^2}{2}x^T\Sigma x$ in (\ref{L3}). More specifically, we solve the following optimization problem
\begin{equation}\label{Qopt}
\begin{split}
 \max_{x} \; \;  & aW_0x^T\gamma - \frac{a^2W_0^2}{2}x^T\Sigma x, \\
s.t.\;\;  & x^T(\mu-r_f\1)=c,\\
\end{split}
\end{equation}
first and plug in the solution, which we denote by $x_c$, into the expression (\ref{L3}) so that the utility maximization problem becomes an optimization problem of a function of one variable $c$. 

\begin{lemma}\label{lem2.6} Consider the optimization problem (\ref{L22}). Let $\bar{x}\in \R^d$ be a solution for this problem. Then $\bar{x}$ solves (\ref{Qopt}) for some $c$.  \end{lemma}
\begin{proof} Define $\bar{c}=:\bar{x}^T(\mu-\l r_f)$. Let $\tilde{x}$ be the solution to the problem (\ref{Qopt}) with $c$ replaced by $\bar{c}$ (here the solution is unique as $\Sigma$ is positive definite by assumption). By the optimality of $\tilde{x}$, we have $g(\bar{x})\le g(\tilde{x})$. Since $\mathcal{L}_Z(s)$ is a decreasing function we have $\mathcal{L}_Z(g(\tilde{x}))\le \mathcal{L}_Z(g(\bar{x}))$. Since $\bar{c}=\bar{x}^T(\mu-\l r_f)=\tilde{x}^T(\mu-\l r_f)$
we have $G(\tilde{x})\le G(\bar{x})$. This shows that $EU(W(\tilde{x}))\geq EU(W(\bar{x}))$. But $\bar{x}$ is optimal for (\ref{L2}) with $D=\R^d$. Therefore we should have $EU(W(\tilde{x}))=EU(W(\bar{x}))$. This implies $G(\tilde{x})=G(\bar{x})$ and this in turn implies 
$g(\bar{x})=g(\tilde{x})$ again due to $\bar{c}=\bar{x}^T(\mu-\l r_f)=\tilde{x}^T(\mu-\l r_f)$. The uniqueness of the optimization point for (\ref{Qopt}) then implies
$\bar{x}=\tilde{x}$. 
\end{proof}

\begin{remark} The Lemma \ref{lem2.6} above gives a characterization of the optimal portfolios for the problem (\ref{L2}). But it doesn't tell us if the optimal portfolio for the problem (\ref{lem2.6}) is unique. It shows only that any optimal portfolio for the problem (\ref{L2}) solves a quadratic optimization problem (\ref{Qopt}) for some appropriate $c$.
 Now consider the case of example \ref{ex2.2}. In the setting of this example, consider the utility maximization problem (\ref{L2}). Since $0\in \R^d$,  as explained in the Example \ref{ex2.2}, the vector $\hat{x}=0$ is the solution for the optimization problem (\ref{L2}). Now let $x^{\star}$ be the optimal solution for 
the problem (\ref{Qopt}) with $c=0$ (which means $(x^{\star})^T(\mu-r_f\1))=0$). Then we should have $g(x^{\star})\ge g(\hat{x})$. But if $g(x^{\star})> g(\hat{x})$, then $\hat{x}=0$ can not be optimal solution for (\ref{L2}). Therefore we should have   $g(x^{\star})=g(\hat{x})$. The uniqueness of the optimal solution for (\ref{Qopt}) with $c=0$ then implies $x^{\star}=\hat{x}=0$.
\end{remark}

\begin{definition}\label{def2.8} Consider the optimization problem 
(\ref{L2}) for some given model (\ref{one}) and for some domain $D\subset \R^d$. Let $\s$ denote the CV-L of the mixing distribution $Z$. Let $x^{\star}\in D$ be a solution for (\ref{L2}). We say that $x^{\star}$ is irregular if $g(x^{\star})=\s$. If $g(x^{\star})>\s$, we call the solution $x^{\star}$  regular. 
\end{definition}

\begin{remark}\label{rem-IN} Clearly the definition of irregular and regular solutions depend on the CV-L number $\s$ of the mixing distribution $Z$ in (\ref{one}).
If $\mathcal{L}_Z(\s)=+\infty$, then the solution to (\ref{L2}) can not be irregular. Therefore, the irregularity can happen only when $\mathcal{L}_Z(\s)<+\infty$. Observe that the solution $x=0$ in Example \ref{ex2.2} is an irregular solution.
\end{remark}

\begin{remark} Consider the optimization problem (\ref{L2}). From Lemma \ref{lem2.6}, any optimal portfolio $x^{\star}$ is a solution for the quadratic optimization problem (\ref{Qopt}) with $x^T(\mu-r_f\1)=c^{\star}$ for some fixed $c^{\star}$. If $x^{\star}$ is irregular, then 
$g(x^{\star})=\s$. The optimality and uniqueness (on the hyperplane $x^T(\mu-r_f\1)=c^{\star}$) of $x^{\star}$ implies that we have $g(x)<g(x^{\star})=\s$ for all $x\neq x^{\star}$ on the hyperplane
$x^T(\mu-r_f\1)=c^{\star}$. Therefore we have $EU(W(x))=-\infty$ for all $x\neq x^{\star}$
on the hyperplane $x^T(\mu-r_f\1)=c^{\star}$. From this we conclude that if the optimal portfolio for the problem  (\ref{Qopt}) is irregular, then any small neighborhood of this portfolio contains some portfolios with infinite expected utility. In comparison, if the optimal portfolio is regular, then it has a small ball around it with finite expected value for each portfolio in this small ball.  
\end{remark}

As it was shown in Lemma \ref{lem2.6}, the solutions of the utility maximization (\ref{L2}) can be obtained by solving the quadratic optimization problem (\ref{Qopt}). For a given optimization problem (\ref{L2}), if we know the corresponding $c$ in (\ref{Qopt}) such that the solution of (\ref{Qopt}) is the solution of (\ref{L2}), then we just need to solve the optimization problem (\ref{Qopt}) to obtain the optimal portfolio.
But figuring out such an $c$ is not a trivial issue. We first prove following Lemma.

\begin{lemma}\label{2.1xc} For any real number $c$, when $x^T(\mu-\1 r_f)=c$, the maximizing point $x_c$ of $g(x)$ is given by
\begin{equation}\label{xc}
x_c=\frac{1}{aW_0}\Big [\Sigma^{-1}\gamma-q_c\Sigma^{-1}(\mu-\1 r_f)\Big ],   
\end{equation}
and we have
\begin{equation}\label{gq}
 g(x_c)=\frac{1}{2}\gamma^T\Sigma^{-1}\gamma-\frac{q^2_c}{2}(\mu-\1 r_f)^T\Sigma^{-1}(\mu-\1 r_f),
\end{equation}
where 
\begin{equation}\label{qcc}
q_c=\frac{\gamma^T\Sigma^{-1}(\mu-\1 r_f)-aW_0c}{(\mu-\1 r_f)^T\Sigma^{-1}(\mu-\1 r_f)}.    
\end{equation}
\end{lemma}
\begin{proof} We form the Lagrangian $L=g(x)+\lambda (c-x^T(\mu-\1 r_f))$ with the  Lagrangian parameter $\lambda$. Denoting the maximizing point by $x_c$, the first order condition gives
\begin{equation}\label{24}
x_c=\frac{1}{aW_0}\Sigma^{-1}\gamma -\frac{\lambda}{a^2W_0^2}\Sigma^{-1}(\mu-\1 r_f).    
\end{equation}
We plug $x_c$ into $x^T_c(\mu-\1 r_f)=c$ and obtain
\begin{equation}
c= \frac{1}{aW_0}\gamma^T \Sigma^{-1}(\mu-\1 r_f)-\frac{\lambda}{a^2W_0^2}(\mu-\1 r_f)^T\Sigma^{-1}(\mu-\1 r_f).    
\end{equation}
From this we find $\lambda$ as follows

\begin{equation}
\lambda= \frac{aW_0\gamma^T \Sigma^{-1}(\mu-\1 r_f)-ca^2W_0^2}{(\mu-\1 r_f)^T\Sigma^{-1}(\mu-\1 r_f)}.    
\end{equation}
Then we plug $\lambda$ into the expression (\ref{24}) of $x_c$ above and obtain (\ref{xc}). To obtain (\ref{gq}), we plug $x_c$ into $g(x)$ in (\ref{Gx}). After doing some algebra we obtain
\begin{equation}
g(x_c)=\frac{1}{2}\gamma^T\Sigma^{-1}\gamma-\frac{1}{2}q_c^2(\mu-\1 r_f)^T\Sigma^{-1}(\mu-\1 r_f), 
\end{equation}
with $q_c$ given as in (\ref{qcc}). This completes the proof.
\end{proof}

For the rest of the paper, as in \cite{Birge-Bedoya}, for convenience, we use the following notations
\begin{equation}\label{ABC}
\mathcal{A}=\gamma^T\Sigma^{-1}\gamma,\; \mathcal{C}=(\mu-\1 r_f)^T\Sigma^{-1}(\mu-\1 r_f),\; \mathcal{B}= \gamma^T \Sigma^{-1}(\mu-\1 r_f).   
\end{equation}
We first observe that $\C>0$ due to the assumption in Remark \ref{rem-mu-r} and the assumption on positive definiteness  of $\Sigma$. With these notations we have
\begin{equation}\label{gqc}
g(x_c)=\frac{\mathcal{A}}{2}-\frac{q_c^2}{2}\mathcal{C}, \; \; q_c=\frac{\mathcal{B}}{\mathcal{C}}-\frac{aW_0}{\mathcal{C}}c.    
\end{equation}

From the relation (\ref{gqc}), we express $c$ as a function of $q_c$
as follows
\begin{equation}\label{qccc}
c=\frac{1}{aW_0}[\mathcal{B}-\mathcal{C}q_c].
\end{equation}

We define the following function 
\begin{equation}\label{H}
Q(\theta)=e^{\mathcal{C}\theta}\mathcal{L}_Z\Big [\frac{1}{2}\mathcal{A}-\frac{\theta^2}{2}\mathcal{C} \Big ],    
\end{equation}
and we define $\tha=:\sqrt{\frac{\mathcal{A}-2\s}{\mathcal{C}}}$, where $\s$ is the IN of $Z$. If $\s=-\infty$,  the $\tha$ is understood to be equal to $+\infty$. Note here that $\s\le 0$ as $Z$ is non-negative random variable. Therefore $\tha$  is well defined. If $\mathcal{L}_Z(\s)<+\infty$, $Q(\theta)$ is finite iff  $\frac{1}{2}\mathcal{A}-\frac{\theta^2}{2}\mathcal{C}\geq \s$ and this translates into: $Q(\theta)$ is finite iff 
$\theta\in [-\tha, \tha]$. If $\mathcal{L}_Z(\s)=+\infty$, $Q(\theta)$ is finite iff $\frac{1}{2}\mathcal{A}-\frac{\theta^2}{2}\mathcal{C}>\s$ and this translates into: $Q(\theta)$ is finite iff $\theta \in (-\tha, \tha)$.

Next we  prove the following Lemma that relates $Q$ to $G$.
\begin{lemma}\label{lem2.14} Let $x_c$ be the solution for the problem (\ref{Qopt}) for a given $c$. Assume $x_c\in S_a$ if $\mathcal{L}_Z(\s)=+\infty$ and $x_c\in \bar{S}_a$ if $\mathcal{L}_Z(\s)<+\infty$. Then for any $x$ with $ x^T(\mu-\1 r_f)=c$, we have
\begin{equation}
e^{-\mathcal{B}}Q(q_c)\le G(x),    
\end{equation}
where $q_c$ is given by (\ref{qcc}) and $\mathcal{B}$ is given  by (\ref{ABC}). We also have $e^{-\mathcal{B}}Q(q_c)=G(x_c)$.
\end{lemma}
\begin{proof} Note that $G(x)=e^{-aW_0x^T(\mu-\1 r_f)}\mathcal{L}_Z(g(x))$. The stated conditions on $x_c$ in the Lemma insures that $G(x_c)=e^{-aW_0c}\mathcal{L}_Z(g(x_c))$ is finite.
Since $g(x)\le g(x_c)$ for any $x$ with 
$x^T(\mu-\1 r_f)=c$ by the definition of $x_c$ (the optimizing point) and also since $\mathcal{L}_Z(s)$ is a decreasing function of $s$ we have
\begin{equation}\label{HG}
G(x_c)\le G(x)
\end{equation}
for any $x$ with $x^T(\mu-\1 r_f)=c$. We plug the $c$ in (\ref{qccc}) into the expression of $G(x_c)$ and obtain
\begin{equation}\label{hqc}
G(x_c)=e^{-\mathcal{B}} e^{\mathcal{C}q_{c}}\mathcal{L}_Z \Big [\frac{1}{2}\mathcal{A}-\frac{q_{c}^2}{2}\mathcal{C} \Big ]=e^{-\mathcal{B}}Q(q_c). 
\end{equation}
\end{proof}

\begin{remark}\label{rem2.6} The above Lemma \ref{lem2.14} shows that the function $G(x)$ achieves its unique (as the solution for (\ref{Qopt}) is unique in a hyperplane) minimum value on the hyperplane $x^T(\mu-r_f\1)=c$ at $x_c$ and its minimum value is given by $e^{\mathcal{-B}}Q(q_c)$ with $q_c$ in (\ref{gqc}). For any $\theta_0\in [-\tha, \tha]$, we can let $c_0$ be such that $q_{c_0}=\theta_0$.
Let $x_0$ be the optimal solution of (\ref{Qopt}) with $c$ replaced by $c_0$. From Lemma \ref{2.1xc}, we have $g(x_0)=\frac{1}{2}\A-\frac{q_{c_0}^2}{2}\C$. If $|q_{c_0}|=\tha$, then $g(x_0)=\s$. If $|q_{c_0}|<\tha$, then 
$g(x_0)>\s$.

\end{remark}

\begin{theorem}\label{2.5} Consider the optimization problem (\ref{L22}). A portfolio  $x^{\star}$ is a solution for (\ref{L22}) if and only if 
\begin{equation}\label{themain}
x^{\star}=\frac{1}{aW_0}\Big [\Sigma^{-1}\gamma -q_{min}\Sigma^{-1}(\mu-\1 r_f)\Big ],    
\end{equation}
for some
\begin{equation}\label{32q}
 q_{min}\in \arg min_{\theta \in \Theta}Q(\theta),
 \end{equation}
where $\Theta=[-\tha, \tha]$ if $\tha=\sqrt{\frac{\mathcal{A}-2\s}{\mathcal{C}}}<\infty$ and $\Theta=(-\infty, +\infty)$ if $\tha=+\infty$. Here $\s$ is the CV-L of the mixing distribution $Z$.
\end{theorem}

\begin{proof} First we show that if
$\hat{x}$ is a  solution for (\ref{L22}), then $\hat{x}$ is given by 
(\ref{themain}). By Lemma \ref{lem2.6}, $\hat{x}$ is a solution for the optimization problem (\ref{Qopt}) with some $c=\hat{c}$. By Lemma  \ref{2.1xc}, $\hat{x}$ takes the following form
\[
\hat{x}=\frac{1}{aW_0}\Big [\Sigma^{-1}\gamma-\hat{q}\Sigma^{-1}(\mu-\1 r_f)\Big ],
\]
with $\hat{q}=\B/\C-(aW_0/\C)\hat{c}$. Again by Lemma \ref{2.1xc} we have (see (\ref{gqc}))
\[
g(\hat{x})=\frac{\A}{2}-\frac{(\hat{q})^2}{2}\C.
\]
Since $\hat{x}$ is a solution for (\ref{L22}) we have $G(\hat{x})<\infty$ and this implies $g(\hat{x})\geq \s$ if $\s$ is finite and $g(\hat{x})>\s$ if $\s=-\infty$ (note that $g(\hat{x})=-\infty$ implies $G(\hat{x})=+\infty$ due to the assumption $Z\neq 0$ in Remark \ref{rem-mu-r} and $G(\hat{x})=e^{-aW_0\hat{x}^T(\mu-r_f\1)}\mathcal{L}_Z(g(\hat{x}))$). The expression of $g(\hat{x})$ above then implies $\hat{q}\in \Theta$ (note here that for the case  $\tha=+\infty$, we can't have $\hat{q}^2=+\infty$ as $g(\hat{x})$ is finite as explained above).

Now we need to show $\hat{q}\in \arg min_{\theta \in \Theta }Q(\theta)$. From Lemma \ref{lem2.14}, we have $G(\hat{x})=e^{-\B}Q(\hat{q})$. Take any $\theta_0\in \Theta$ (including the case $\Theta=(-\infty, +\infty)$). Let $c_0$ be such that $\theta_0=q_{c_0}$ (see Remark \ref{rem2.6}). Let $x_0$ be the solution for (\ref{Qopt}) with $c$ replaced by $c_0$. By Lemma \ref{2.1xc} we have $g(x_0)=\frac{\A}{2}-\frac{(q_{c_0})^2}{2}\C$. Since $\theta_0=q_{c_0}\in \Theta$, we have $g(x_0)\geq \s$ if $\s$ is finite and $g(x_0)> \s$ if $\s=-\infty$. Therefore either $x_0\in S_a$ or $x_0\in \bar{S}_a$. Then by Lemma \ref{lem2.14} we have $G(x_0)=e^{-\B}Q(q_{c_0})$. Since $\hat{x}$ is the optimal portfolio it is the minimizing point for the function $G(x)$ (see (\ref{arG}) for this). Therefore we have $G(\hat{x})\le G(x_0)$. This implies $Q(\hat{q})\le Q(q_{c_0})=Q(\theta_0)$. Since $\theta_0$ is arbitrary, we conclude that $\hat{q}\in \arg min_{\theta \in \Theta }Q(\theta)$.

Next we show that any portfolio of the form (\ref{themain}) is an optimal portfolio for (\ref{L22}). Fix an arbitrary $q_m\in \arg min_{\theta \in \Theta}Q(\theta)$. Then $q_m\in [-\tha, \tha]$ if $\tha$ is finite and $q_m\in (-\infty, +\infty)$ if $\tha=+\infty$. Let $c_m$ be such that $q_m=q_{c_m}$ and let $x_m$ be the solution of (\ref{Qopt}) with $c$ replaced by $c_m$. By Lemma \ref{2.1xc},  we have 
\[
x_m=\frac{1}{aW_0}\Big [\Sigma^{-1}\gamma -q_m\Sigma^{-1}(\mu-\1 r_f)\Big ],  
\]
and $g(x_m)=\frac{\A}{2}-\frac{q_m^2}{2}\C$. The condition on $q_m$ above implies $g(x_m)\geq \s$ if $\s$ is finite and $g(x_m)>-\infty$ if $\s=-\infty$. Therefore either $x_m\in S_a$ or $x_m\in \bar{S}_a$. By Lemma \ref{lem2.14} we have $G(x_m)=e^{-\B}Q(q_m)$ which is a finite number. To show $x_m$ is an optimal portfolio we need to show $G(x_m)\le G(x)$ for any $x$ that $G(x)$ is finite (note that either $G(x)=+\infty$ or it is finite). Fix an arbitrary $\bar{x}$ with $G(\bar{x})<+\infty$. Let $\bar{c}=\bar{x}^T(\mu-r_f\1)$. Let $x_{\bar{c}}$ be the solution of (\ref{Qopt}) with $c$ replaced by $\bar{c}$. Since $G(x)<\infty$ we either have $x\in \bar{S}_a$
or $x\in S_a$. This means that $x_{\bar{c}}\in \bar{S}_a$.  By Lemma \ref{2.1xc} we have 
$g(x_{\bar{c}})=\frac{\A}{2}-\frac{q_{\bar{c}}^2}{2}\C$, where $q_{\bar{c}}$ is given by (\ref{gqc}) with $c$ replaced by $\bar{c}$. Therefore we have $q_{\bar{c}}\in [-\tha, \tha]$ if $\tha$ is finite and $q_{\bar{c}}\in (-\infty, +\infty)$ if $\tha=+\infty$.  By the definition of $q_m$, we have $Q(q_m)\le Q(q_{\bar{c}})$. Therefore we have 
$G(x_m)=e^{-\B}Q(q_m)\le e^{-\B}Q(q_{\bar{c}})=G(\bar{x})$.
\end{proof}

\begin{proposition}\label{2.55} Consider the optimization problem (\ref{L22}). If  $x^{\star}$ is a regular solution for (\ref{L22}) then
\begin{equation}\label{themainn}
x^{\star}=\frac{1}{aW_0}\Big [\Sigma^{-1}\gamma -q_{min}\Sigma^{-1}(\mu-\1 r_f)\Big ],    
\end{equation}
for some
\begin{equation}\label{32qq}
 q_{min}\in \arg min_{\theta \in (-\tha, \tha)}Q(\theta),
 \end{equation}
where $\tha=:\sqrt{\frac{\mathcal{A}-2\s}{\mathcal{C}}}$ and $\s$ is the CV-L of the mixing distribution $Z$. 
\end{proposition}

\begin{proof} Let $\hat{x}$ be a regular solution. By Lemma \ref{lem2.6}, $\hat{x}$ is a solution for the optimization problem (\ref{Qopt}) with some $c=\hat{c}$. By Lemma  \ref{2.1xc}, $\hat{x}$ takes the following form
\[
\hat{x}=\frac{1}{aW_0}\Big [\Sigma^{-1}\gamma-\hat{q}\Sigma^{-1}(\mu-\1 r_f)\Big ],
\]
with $\hat{q}=\B/\C-(aW_0/\C)\hat{c}$. Again by Lemma \ref{2.1xc} we have (see (\ref{gqc}))
\[
g(\hat{x})=\frac{\A}{2}-\frac{(\hat{q})^2}{2}\C.
\]
Since $\hat{x}$ is regular we have $g(\hat{x})>\s$. From this we conclude $\hat{q}\in (-\tha, \tha)$. From Lemma \ref{lem2.14}, we have $G(\hat{x})=e^{-\B}Q(\hat{q})$. Note that $\hat{q}=q_{\hat{c}}$.  Now we show that $\hat{q}\in \arg min_{\theta \in (-\tha, \tha)}Q(\theta)$. Take any $\theta_0\in (-\tha, \tha)$. Let $c_0$ be such that $\theta_0=q_{c_0}$ (see Remark \ref{rem2.6}). Let $x_0$ be the solution for (\ref{Qopt}) with $c$ replaced by $c_0$. By Lemma \ref{2.1xc} we have $g(x_0)=\frac{\A}{2}-\frac{(q_{c_0})^2}{2}\C$. Since $\theta_0=q_{c_0}\in (-\tha, \tha)$, we have $g(x_0)>\s$. Therefore $x_0\in S_a$. Then by Lemma \ref{lem2.14} we have $G(x_0)=e^{-\B}Q(q_{c_0})$. Since $\hat{x}$ is the optimal portfolio it is the minimizing point for the function $G(x)$ (see (\ref{arG}) for this). Therefore we have $G(\hat{x})\le G(x_0)$. This implies $Q(\hat{q})\le Q(q_{c_0})=Q(\theta_0)$. Since $\theta_0$ is arbitrary, we conclude that $\hat{q}\in \arg min_{\theta \in (-\tha, \tha)}Q(\theta)$.
\end{proof}

\begin{remark} Let us look at the case of Example \ref{ex2.2}. From the analysis in this example the optimal solution for the problem (\ref{L22}) is $x^{\star}=0$ and it is unique. Here we would like to check that this optimal portfolio $x^{\star}=0$ can also be derived from (\ref{themain}). To see this, note that in this example $\gamma=0$. Therefore we have $Q(\theta)=e^{\mathcal{C}\theta}\mathcal{L}_Z(-\frac{\theta^2}{2}\mathcal{C})$ and $q_c=-\frac{aW_0}{\mathcal{C}}c$. Observe that $0\in \{x^T(\mu-\1r_f): x\in \R^n\}$. Also for any $\theta \neq 0$ we have $Q(\theta)=+\infty$ as the CV-L of $Z\sim e^{N(0, 1)}$ is 
 $\s=0$.
Therefore $\arg min_{\theta \in \Theta }Q(\theta)$ has only one element $q_{min}=0$. Then (\ref{themain}) gives $\bar{x}^{\star}=0$ as the only optimal solution. Observe that in fact in this example we have $\A=0$ and therefore $\tha=0$. Thus $\bar{q}_{min}=\arg min_{\theta\in \{0\}}Q(\theta)=0$.
\end{remark}

\section{Large financial markets}

In the previous section we gave closed form solution for the optimal portfolio for an exponential utility maximizer in a market that contains one risk-free asset and  finitely many risky assets with return vector that follow (\ref{one}). Our Theorem \ref{2.5} gives complete characterization 
of the optimal portfolio in such small markets. 

The next natural question to ask is what happens if the consumer with exponential utility wants to increase her expected utility as much as possible by adding as many as necessary assets into her portfolio. We can best investigate this possibility
by working in mathematical models with countably infinitely
many assets. 



In this section we consider a sequence of economies with increasing number of assets. In the $n$th economy, there are $n$ risky assets and one
riskless asset. The return vector of the risky assets in the $n$th economy satisfies (\ref{one}). A consumer with exponential utility maximizes her expected utility based on the $n+1$ assets in each $n$th economy. Our main concern in this section is to investigate if the optimal expected utility of the consumer converges to a limit as $n\to\infty$
and we would like to identify this limit as
the optimizer in the market with infinitely
many assets.

Such ``stability'' of optimal investment problems
was proved in
\cite{laurence-Miklos-2021} for a wide range of
models. The methods of \cite{laurence-Miklos-2021},
however, cannot deal with exponential utilities.
So we need to apply somewhat different, new
arguments.

Our main result in this section shows that the consumer can achieve the maximum possible (in a market where she can trade on countably infinite risky assets) expected utility
by following the sequence of  optimal trading strategies in each $n$th economy, which are shown to converge to a limit (see our Lemma \ref{exi} below). We call this limit portfolio the 
``overall best optimal portfolio'' in this paper.

An economy that allows to trade on countably infinite  risky assets is called a large financial market in the literature. They serve well to
describe e.g.\ bonds of various maturities.
A first model of this type, the ``Arbitrage Pricing Model'' (APM) goes back to
\cite{Ross-1976}. We consider a slight extension of that model in the present section. As the main result of this section, we will show that the exponential utiliy maximization problem in large financial market  can be approximated by similar problems for finitely many assets (and the latter can be solved by the results of the 
previous sections).

Before we state and prove our main result of this section, we first specify the structure of our $n$th economy for all $n$. Return on the bank account is $R_{0}:=r_{f}$ where $r_{f}\geq 0$ is the risk-free interest rate.
For simplicity we assume $r_{f}=0$ henceforth.
For $i=1$, $R_{1} := \gamma_{1}Z+\mu_{1}+\bar{\beta}_1\sqrt{Z}\varepsilon_1$ is the return on the ``market portfolio'',
which may be thought of as an investment into an index.
For $i\geq 2$, let the return on risky asset $i$ be given by
\begin{eqnarray}\label{returnss}
R_i = \gamma_{i}Z+\mu_{i}+\beta_i\sqrt{Z}\varepsilon_1+\bar{\beta}_i\sqrt{Z}\varepsilon_i.
\end{eqnarray}
Here the $(\varepsilon_i)_{i \geq 1}$ are assumed independent standard Gaussian, $Z$ is a positive random variable,
independent of the $\varepsilon_{i}$, $\beta_{i}$, $i\geq 2$, $\bar{\beta}_{i}\neq 0,\gamma_{i},\mu_{i}$, 
$i\geq 1$ are constants.
The classical APM corresponds to $Z\equiv 1$. 
We refer to \cite{Ross-1976} for further discussions on that model.

We consider
investment strategies in finite market segments. A strategy investing in the first $n$
assets is a sequence of numbers $\phi_{0},\phi_{1},\ldots,\phi_{n}$. For simplicity, we assume $0$ initial
capital and also that every asset has price $1$ at time $0$. Self-financing imposes $\sum_{i=0}^{n}\phi_{i}=0$
so a strategy is, in fact, described by $\phi_{1},\ldots,\phi_{n}$ which can be \emph{arbitrary} real numbers.
The return on the portfolio $\phi$ is thus
$$
V(\phi)=\sum_{i=1}^{n}\phi_{i}R_{i},
$$
noting also that $R_{0}=0$ is assumed. 

For utility maximization to be well-posed, one should assume a certain arbitrage-free property for the market.
Notice that a probability $Q_{n}\sim P$ is a martingale measure 
for the first $n$ assets (that is, $E_{Q_{n}}[R_{i}]=0$ for all $1\leq i\leq n$)
provided that
\begin{equation}\label{melo1}
E_{Q_{n}}[\varepsilon_{1}|Z=z]=b_{1}(z):=-\frac{\gamma_{1}\sqrt{z}}{\bar{\beta}_1}-\frac{\mu_{1}}{\sqrt{z}\bar{\beta}_1},\ z\in (0,\infty) 
\end{equation}
and, for each $i\geq 2$,
\begin{equation}\label{melo2}
E_{Q_{n}}[\varepsilon_{i}|Z=z]=b_{i}(z):=-\frac{\gamma_{i}\sqrt{z}}{\bar{\beta}_i}-\frac{\mu_{i}}{\sqrt{z}\bar{\beta}_i}
-\frac{\beta_{i}b_{1}(z)\sqrt{z}}{\bar{\beta}_{i}}
,\ z\in (0,\infty).
\end{equation}

Now notice that, in fact, 
the set of such $V(\phi)$ coincides with the set of $$
V(h):=\sum_{i=1}^{n}h_{i}\sqrt{Z}(\varepsilon_{i}-b_{i}(Z))
$$
where $h_{1},\ldots,h_{n}$ are arbitrary real numbers. We denote by $H_{n}$ the set of all $n$-tuples $(h_{1},\ldots,h_{n})$.
It is more convenient to use this ``$h$-parametrization'' in the sequel.

\begin{assumption}\label{zbounded}
There are $0<c<C$ such that $c\leq Z\leq C$.	
\end{assumption}

Let us define $d_{i}:=\sup_{z\in [c,C]}|b_{i}(z)|$, $i\geq 1$.
The next assumption is similar in spirit to the no-arbitrage condition derived in \cite{Ross-1976}, see also \cite{Miklos-2004}.

\begin{assumption}\label{csummable} We stipulate
$
\sum_{i=1}^{\infty}d_{i}^{2}<\infty.
$
\end{assumption}

\noindent\textbf{Fact.} If $X$ is standard normal then $E[e^{-\theta X-\theta^{2}/2}]=1$ and
$E[Xe^{-\theta X-\theta^{2}/2}]=\theta$, for all $\theta\in\mathbb{R}$.
Notice also that, for all $p\geq 1$,
\begin{equation}\label{negyzet}
E[e^{-p\theta X-p\theta^{2}/2}]=e^{(p^{2}-p)\theta^{2}/2}.
\end{equation}

Let us now define 
$$
f_{n}(z):=\exp\left(-\sum_{i=1}^{n}[b_{i}(z)\varepsilon_{i}+b_{i}(z)^{2}]\right).
$$

Clearly, $E[f_{n}(z)]=1$ and $E[f_{n}(z)\varepsilon_{i}]=b_{i}(z)$ for $i=1,\ldots,n$. Then $Q_{n}$
defined by $dQ_{n}/dP:=f_{n}(Z)$ will be a martingale measure for the first $n$ assets. Indeed,
$$
E[f_{n}(Z)]=\int_{[c,C]}E[f_{n}(z)]\, \mathrm{Law}(Z)(dz)=1
$$
and
$$
E[f_{n}(Z)\varepsilon_{i}|Z=z]=E[(\varepsilon_{i}-b_{i}(z))e^{-b_{i}(z) \varepsilon_{i}-b_{i}(z)^{2}/2}]=0,\ 1\leq i\leq n.
$$
It follows from \eqref{negyzet} and from Assumption \ref{csummable} that $\sup_{n}E[(dQ_{n}/dP)^{2}]<\infty$ hence
$dQ/dP:=\lim_{n\to\infty}dQ_{n}/dP$ exists almost surely and in $L^{2}$, and this is a martingale
measure for \emph{all} the assets, that is, $E_{Q}[R_{i}]=0$ for all $i\geq 1$. 
Note also that $E[(dQ/dP)^{2}]<\infty$.

Using the previous sections, we may find $h^{*}_{n}\in H_{n}$ such that
$$
U_{n}:=E[e^{-V(h^{*}_{n})}]=\min_{h\in H_{n}}E[e^{-V(h)}].
$$
If we wish to find (asymptotically) optimal strategies for this large financial market then we also
need to verify that $U_{n}\to U:=\inf_{h\in \cup_{n\geq 1}H_{n}}E[e^{-V(h)}]$ as $n\to\infty$.{}

Let us introduce 
$$\ell_2:=\left\{(h_i)_{i\geq 1}, \, h_{i}\in\mathbb{R},\, i\geq 1,\,  \sum_{i=1}^{\infty}h_i^2<\infty\right\}$$ 
which is a Hilbert space with the norm $||h||_{\ell_2}:=\sqrt{\sum_{i=1}^{\infty}h_i^2}$.
We may and will identify each $(h_{1},\ldots,h_{n})\in H_{n}$ with $(h_{1},h_{2},\ldots)\in\ell_{2}$ for all $n\geq 1$.
Also define $d:=(d_{1},d_{2},\ldots)\in\ell_{2}$. 

\begin{theorem}\label{mainu} Under Assumptions \ref{zbounded} and \ref{csummable}, one has $U_{n}\to U$, $n\to\infty$.{}
\end{theorem}
\begin{proof}
It follows from Lemma \ref{exi} below that there is $\bar{h}^{*}\in\ell_{2}$ such that $U=E[e^{-V(\bar{h}^{*})}]$.
Define now $\tilde{h}_{n}:=(\bar{h}^{*}_{1},\ldots,\bar{h}^{*}_{n})\in H_{n}$.
It is clear that $U_{n}\geq U$ and $E[e^{-V(\tilde{h}_{n})}]\geq U_{n}$ for all $n\geq 1$.
Hence it remains to establish $E[e^{-V(\tilde{h}_{n})}]\to U$. 

Noting that $V(\tilde{h}_{n})\to V(h^{*})$ almost surely, 
it suffices to show that $\sup_{n\in\mathbb{N}}E[e^{-2V(\tilde{h}_{n})}]<\infty$.  	
This follows from 
\begin{eqnarray*}
E[e^{-2V(\tilde{h}_{n})}]\leq e^{2\sqrt{C}||\tilde{h}_{n}||_{2}||d||_{2}}E[e^{2\sqrt{C}||\tilde{h}_{n}||_{2}|N|}]
\leq e^{2\sqrt{C}||h^{*}||_{2}||d||_{2}}E[e^{2\sqrt{C}||h^{*}||_{2}|N|}],	
\end{eqnarray*}
where $N$ is a standard normal random variable.
\end{proof}

\begin{lemma}\label{rexi}
There exists $\alpha>0$, such that for all $h \in \ell_2$ with $\|h\|_{\ell_2}=1,$ 
 $P(V(h)\leq -\alpha) \geq \alpha$ holds.
\end{lemma}
\begin{proof} We follow closely the proof of Proposition 3.2 in \cite{laurence-Miklos-2021}, see also
\cite{carassus2020risk}. 
We argue by contradiction.
Assume that for all $n\geq 1$, there is $g_{n}=(g_{n}(1),g_{n}(2),\ldots)\in \cup_{n\geq 1}H_{n}$ with $\|g_{n}\|_{\ell_2}=1$ and 
$P\left(V(g_{n}) \leq-{1}/{n}\right) \leq {1}/{n}$. \\Clearly, $V(g_{n})^{-}\to 0$ in probability as $n\to\infty$. 
We claim that $E_{Q}[V(g_{n})^{-}]\to 0$. By the Cauchy-Schwarz inequality 
\begin{eqnarray*}
E_{Q}[V(g_{n})^{-}] & \leq &  \|dQ/dP\|_{L^2(P)} \left(E[(V(g_{n})^-)^{2}]\right)^{1/2}.
\end{eqnarray*}
However, \begin{equation}\label{rofa}
V(g_{n})^{-}\leq |V(g_{n})|\leq \sqrt{C}[|N|+||d||_{2}]
\end{equation} 
for some standard normal $N$. This implies $E[(V(g_{n})^-)^{2}]$, $n\to\infty$ and hence our claim.

Since
$E_{Q}[V(g_{n})]=0$ by the martingale measure property of $Q$, we also get that $E_{Q}[V(g_{n})^{+}]\to 0$. 
It follows that
$E_{Q}[|V(g_{n})|]\to 0,$
hence $V(g_{n})$ goes to zero $Q$-a.s. (along a subsequence) and, as $Q$ is equivalent to $P$, $P$-a.s.  
Using that $|V(g_{n})|^{2}$, $n\in\mathbb{N}$ is 
uniformly $P$-integrable by \eqref{rofa}, we get $E[V(g_{n})^{2}]\to 0.$ 
An auxiliary calculation gives
$$
E[V(g_{n})^{2}]=\Vert g_{n}\Vert_{\ell_2}^{2}E[Z]+ \sum_{i=1}^{\infty}g^{2}_{n}(i)E[b_{i}^{2}(Z)Z]\geq E[Z]>0,
$$
a contradiction showing our lemma.
\end{proof}

\begin{lemma}\label{exi}
There is $h^{*}\in\ell_{2}$ such that $U=E[e^{-V(h^{*})}]$.	
\end{lemma}
\begin{proof}
There are $h_{n}\in \cup_{j\in\mathbb{N}}H_{j}$, $n\in\mathbb{N}$ such that
$E[e^{-V(h_{n})}]\to U$. If we had $\sup_{n}||h_{n}||_{\ell_{2}}=\infty$ then
(taking a subsequence still denoted by $n$), $||h_{n}||_{\ell_{2}}\to\infty$, $n\to\infty$.
By Lemma \ref{rexi},  $$P(V(h_{n})\leq -\alpha ||h_{n}||_{\ell_{2}})\geq \alpha$$
and this implies $E[e^{-V(h_{n})}]\to\infty$, which contradicts $E[e^{-V(h_{n})}]\to U\leq E[e^{0}]=1$.

Then necessarily $\sup_{n}||h_{n}||_{\ell_{2}}<\infty$ and the Banach-Saks theorem
implies that convex combinations $\bar{h}_{n}$ of the $h_{n}$ converge to some $h^{*}\in\ell_{2}$
(in the norm of $\ell_{2}$).
By Fatou's lemma, 
$$
E[e^{-V(h^{*})}]\leq \liminf_{n\to\infty}E[e^{-V(\bar{h}_{n})}]\leq \liminf_{n\to\infty}E[e^{-V({h}_{n})}]
= U,$$
using also convexity of the exponential function. This proves the statement.
\end{proof}

\section{Applications and examples}

Our Theorem \ref{2.5} gives closed form expression for the optimal portfolios for the problem (\ref{L22}) by using the function $Q(\theta)$ defined in (\ref{H}). In this section, we first study some properties of this function. Then we present some examples. 

Let $\M_Z(s)=Ee^{sZ}$ and $\K_Z(s)=\ln \M_Z(s)$ denote the moment generating function (MGF) and the cumulant generating function (KGF) of the mixing distribution $Z$ respectively. We have the following obvious relation
\[
Q(\theta)=e^{\C\theta}\M_Z(\frac{\C}{2}\theta^2-\frac{\A}{2}), \; \; \ln Q(\theta)=\C\theta+\K_Z(\frac{\C}{2}\theta^2-\frac{\A}{2}). 
\]
Therefore the minimizing points of $Q(\theta)$ in (\ref{32q}) can also be found by using the MGF or KGF of $Z$. In the following Lemma we state some properties of the function $Q(\theta)$.

\begin{lemma}\label{HH} Consider the model (\ref{one}) with a non-trivial mixing distribution $Z$. Let $\s$ denote the CV-L of $Z$ and $\tha$ is defined as in Section 2. Let the function $Q(\theta)$ be defined by (\ref{H}). Assume our model (\ref{one}) is such that either $\A\neq 0$ or  $\s\neq 0$ which insures $\tha=\sqrt{(\A-2\s)/\C}\neq 0$ and hence $(-\tha, \tha)$ is a non-empty open interval. Then we have the following.
\begin{itemize}
\item [a)] The function $Q(\theta)$ is infinitely differentiable on $(-\tha, \tha)$. If $\s$ is finite and $\LL_Z(\s)=+\infty$ or if $\s=-\infty$, we have
\begin{equation}\label{Qinf}
\lim_{\theta \rightarrow \tha^-}Q(\theta)=+\infty,\;\;\; \lim_{\theta \rightarrow -\tha^+}Q(\theta)=+\infty.
\end{equation}
When $\s$ is finite and $\LL_Z(\s)<\infty$ we have $Q(\tha)<\infty$ and $Q(-\tha)<\infty$. When $\s$ is finite and $\theta \notin [-\tha, \tha]$ we have $Q(\theta)=+\infty$.

\item [b)] The function $Q(\theta)$ is  strictly increasing on $[0, \tha]$ when $\s$ is finite. It is
 strictly increasing on $[0, +\infty)$ when $\s=-\infty$. We have $Q'(0)\neq 0$  which implies the $q_{min}$ in (\ref{themain}) can not be zero under the stated conditions in the Lemma.

\item[c)] The function $Q(\theta)$ is strictly convex on the open interval $(-\tha, \tha)$ when $\s$ is finite and $\LL(\s)=+\infty$ or when $\s=-\infty$. $Q(\theta)$ is strictly convex on $[-\tha, \tha]$ when $\s$ is finite and $\LL(\s)<\infty$. 
\end{itemize}

\end{lemma}
\begin{proof} a) It is sufficient to prove that the function $\theta \rightarrow \LL_Z(\frac{\A}{2}-\frac{\C}{2}\theta^2)$ is infinitely differentiable when $\theta \in (-\tha, \tha)$. This function is a composition of two functions $s\rightarrow \LL_Z(s)$ and $\theta \rightarrow  \frac{\A}{2}-\frac{\C}{2}\theta^2$. So it is sufficient to prove the infinite differentiability of $s\rightarrow \LL_Z(s)$ in the corresponding domain. If $\LL_Z(s)$ is $k$'th order differentiable then we would have $\LL_Z^{(k)}(s)=(-s)^kE[Z^ke^{-sZ}]$. To justify the change of the order of derivative with expectation for this we need to show $E[Z^ke^{-sZ}]<\infty$. Let us look at the case $\s\neq 0$ first. In this case  we have $Ee^{sZ}<\infty$ in $(-\infty, |\s|)$. Thus all the moments of $Z$ are finite. This implies $E[Z^ke^{-sZ}]<\infty$ for any positive integer $k$ and all $s\in (\s, +\infty)$. If $\theta\in (-\tha, \tha)$, then $\frac{\A}{2}-\frac{\C}{2}\theta^2 \in (\s, \frac{\A}{2})$. Therefore when $\s\neq 0$, the infinite differentiability of $Q(\theta)$ follows. Now let us look at the case $\s=0$. In this case $\tha=\sqrt{\frac{\A}{\C}}$ and for any $\theta\in (-\tha, \tha)$ we have $\frac{\A}{2}-\frac{\C}{2}\theta^2\in (0, \frac{\A}{2})$. Therefore it is sufficient to prove infinite differentiability of $\LL_Z(s)$ on $(0, \frac{\A}{2})$. Fix an arbitrary positive integer $k$.
When $s\in (0, \frac{\A}{2})$ we have $Z^k/e^{sZ}=(Z^k/e^{sZ})1_{\{Z\le M\}}+(Z^k/e^{sZ})1_{\{Z>M\}}$ for any positive number $M$. For sufficiently large $M=M_0$, we have $(Z^k/e^{sZ})1_{\{Z>M_0\}}\le 1$ and  $Z^k/e^{sZ}=(Z^k/e^{sZ})1_{\{Z\le M_0\}}$ is a bounded random variable. Thus $E(Z^ke^{-sZ})<\infty$ for any positive integer $k$ when  $s\in (0, \frac{\A}{2})$. This shows that 
$\theta \rightarrow  \LL_Z(\frac{\A}{2}-\frac{\C}{2}\theta^2)$ is infinitely differentiable when $\s=0$ also.

When $\s$ is finite and when $\theta \rightarrow \tha$ from the left-hand-side or when $\theta \rightarrow -\tha$ from the right-hand-side, the function  $\frac{\A}{2}-\frac{\C}{2}\theta^2$ decreasingly converges to $\s$ (in some neighborhood of $\s$). Then the monotone convergence theorem gives the claim (\ref{Qinf}). 
Now assume $\s=-\infty$ which happens when the mixing distribution $Z$ is a bounded non-trivial random variable. The result $\lim _{\theta \rightarrow +\infty}Q(\theta)=+\infty$ is clear as both $e^{\C \theta}$ and $\LL_Z(\frac{\A}{2}-\frac{\theta^2}{2}\C)$  go to $+\infty$. The limit $\lim _{\theta \rightarrow -\infty}Q(\theta)=+\infty$ is less clear as $e^{\C \theta}\rightarrow 0$ and $\LL_Z(\frac{\A}{2}-\frac{\theta^2}{2}\C)\rightarrow +\infty$ in this case. But since $Z\neq 0$ with positive probability, we have a positive number $\delta>0$ with $P(Z\geq \delta)>0$. We have the following
\begin{equation}\label{Qdelta}
Q(\theta)=Ee^{[\frac{\C}{2}\theta^2-\frac{\A}{2}]Z+\C\theta}\geq e^{[\frac{\C}{2}\theta^2-\frac{\A}{2}]\delta +\C\theta}P(Z\geq \delta),
\end{equation}
for all $\theta$ with $\frac{\C}{2}\theta^2-\frac{\A}{2}>0$. Then, since the right-hand-side of (\ref{Qdelta}) goes to $+\infty$ when $\theta \rightarrow -\infty$, the claim follows. The remaining property of $Q$ in part a) above is obvious by the definition of $\tha$.

b) For any $\theta\in (-\tha, \tha)$ we have
\begin{equation}\label{FOCQ}
Q'(\theta)=\C e^{\C\theta}\mathcal{L}_Z[\frac{\A}{2}-\frac{\theta^2}{2}\C]-\theta \C e^{\C\theta}\mathcal{L}'_Z[\frac{\A}{2}-\frac{\theta^2}{2}\C].
\end{equation}
Observe that $0\in (-\tha, \tha)$ always (in both cases $\s\neq 0$ and $\s=0$). Therefore $Q'(0)$ always exists and from (\ref{FOCQ}) we see that $Q'(0)\neq 0$. Now since $\mathcal{L}_Z(s)$ is a strictly decreasing function  we have $\mathcal{L}'_Z(s)<0$. Therefore 
$Q'(\theta)$ is finite and 
$Q'(\theta)>0$ when $\theta \in (0, \tha)$. At $\theta=0$, we have $Q(0)=\C \LL_Z(\A/2)$ and clearly we have $Q(0)<Q(\theta)$ for all $\theta \in (0, \tha)$. At $\theta=\tha$, we have $Q(\theta)= \LL_Z(\s)$ which is either $+\infty$ or finite. When it is finite we have $Q(\theta)<Q(\tha)$ for all $\theta \in [0, \tha)$ also. 

c) Define $f_z(\theta)=:e^{\frac{\C}{2}z\theta^2+\C\theta-\frac{\A}{2}z}$ for any real number $z\geq 0$ and for all $\theta \in \R$. We have $f'_z(\theta)=(\C z\theta+\C)e^{\frac{\C}{2}z\theta^2+\C \theta-\frac{\A}{2}z}$ and $f^{''}_z(\theta)=\C ze^{\frac{\C}{2}z\theta^2+\C\theta-\frac{\A}{2}z}+(\C z \theta+\C)^2e^{\frac{\C}{2}z\theta^2+\C\theta-\frac{\A}{2}z}>0$ for any $z\geq 0$. Therefore $f_z(\theta)$ is a strictly convex function for any fixed $z\geq 0$. Therefore we have
\[
f_z(\lambda \theta_1+(1-\lambda)\theta_2)<\lambda f_z(\theta_1)+(1-\lambda)f_z(\theta_2)
\]
for any $\lambda \in [0, 1]$ and for all $\theta_1, \theta_2 \in \R$ for each fixed $z\geq 0$. This strict inequality also holds when $z=Z$. Also observe that when $\s$ is finite and $\LL_Z(\s)=+\infty$ or when $\s=-\infty$,  for $\theta_1, \theta_2 \in (-\tha, \tha)$ we have $Ef_Z(\theta_1)<\infty$ and $Ef_Z(\theta_2)<\infty$.
When $\s$ is finite and $\LL_Z(\s)<\infty$, for all $\theta_, \theta_2 \in [-\tha, \tha]$ we have  $Ef_Z(\theta_1)<\infty$ and $Ef_Z(\theta_2)<\infty$.
We take expectation to the above inequality when $z=Z$ and obtain $Q(\lambda \theta_1+(1-\lambda)\theta_2)<\lambda_1 Q(\theta_1)+(1-\lambda)Q(\theta_2)$. This shows the strict convexity of $Q(\theta)$ that is stated in the Lemma.
\end{proof}

\begin{remark}\label{4.2} The main message of the above Lemma \ref{HH} is that the optimal solution for the problem (\ref{L22}) is always unique. Now assume $\mathcal{L}_Z(\s)<\infty$. In this case, if the optimal portfolio $x^{\star}$ for the problem (\ref{L22}) is irregular then the $q_{min}$ in (\ref{themain}) satisfy $q_{min}=-\tha$. This means that $-\tha$ is the minimizing point of $Q(\theta)$ in $[-\tha, \tha]$. As $Q(\theta)$ is a strictly convex function on $[-\tha, \tha]$ as shown in the above Lemma \ref{HH}, we conclude that $Q(\theta)$ is a strictly increasing, strictly convex function on $[-\tha, \tha]$. In comparision, when the solution for (\ref{L22}) is regular, then the corresponding $Q(\theta)$ is strictly convex but not strictly increasing on $[-\tha, \tha]$.
\end{remark}

\begin{example} \label{5.3ex} Assume the mixing distribution $Z$ in our model (\ref{one}) takes finitely many values 
$\{z_i\}_{1\le i\le  m}$ with corresponding probabilities $(p_i)_{1\le i \le m}$. Then $X$ in (\ref{one}) is a mixture of Normal
random vectors
\begin{equation}
X\sim \sum_{i=1}^mp_iN_d(\mu+\gamma z_i, z_i\Sigma).    
\end{equation}
In this case, the function $Q(\theta)$ takes the following form
\begin{equation}
Q(\theta)=\sum_{i=1}^mp_ie^{(\frac{\theta^2}{2}C-\frac{1}{2}A)z_i+\theta C}.    
\end{equation}
From part c) of the above Lemma \ref{HH} we know that the function $Q(\theta)$ is strictly convex on $(-\infty, +\infty)$. Thus the solution for the optimization problem (\ref{L22}) is unique and this unique solution is given by (\ref{themainn}) with $q_{min}=\arg min_{\theta \in (-\infty, 0)} Q(\theta)$. Now, assume  $Z=1$ with probability one instead. Then $\mathcal{L}_Z(s)=e^{-s}$ and in this case it is easy to see that
\[
Q(\theta)=e^{\frac{C}{2}(\theta^2+2\theta)-\frac{A}{2}}.
\]
The minimizing point of this function is $\theta=-1$ and so $q_{min}=-1$. Then, from (\ref{themain}), the optimal portfolio is given by
\[
x^{\star}=\frac{1}{aW_0}\Sigma^{-1}(\gamma +\mu-\1 r_f).
\]
Note here that  since we assumed $Z=1$, the $X$ in (\ref{one}) is a  Gaussian random vector and therefore one can obtain the above optimal portfolio by direct calculation as our utility function is exponential. However, our above approach seems more convenient.
\end{example}

In the next example, we look at the case of GH models.
\begin{example} \label{GIG-ab}  Lets look at the case of the model (\ref{one}) when the mixing distribution $Z$ is given by GIG models. First assume $Z\sim iG(\lambda, \frac{a^2}{2})$, the inverse Gaussian distribution. In this case we have $\lambda<0$ by the definition of inverse Gaussian random variable. From Proposition 9 of \cite{Hammerstein_EAv_2010} we have $\LL_Z(s)=(\frac{2}{a \sqrt{2s}})^{\lambda}\frac{2K_{\lambda}(a \sqrt{2s})}{\Gamma(-\lambda)}$ and therefore $Q(\theta)=e^{\C \theta}(\frac{2}{a \sqrt{\A-\C \theta^2}})^{\lambda}\frac{2K_{\lambda}(a \sqrt{\A-\C\theta^2})}{\Gamma(-\lambda)}$. In this case, the CV-L is $\s=0$ and $\tha=\sqrt{\A/\C}$. If $\gamma=0$, as discussed in the Example \ref{ex2.2},  the optimal solution for (\ref{L22}) is $x^{\star}=0$. In this case, this solution $x^{\star}=0$ is an irregular solution. Note that in this case $\A=0$ and therefore $\tha=0$. If $\gamma \neq 0$, then $\tha>0$ and in this case the $q_{min}$ in (\ref{themain}) is given by $q_{min}=\arg min_{\theta\in [-\sqrt{\A/\C}, 0)} Q(\theta)$ (due to Lemma \ref{HH}). Note that either by using the fact $\s=0$ or by using the property (A. 8) in \cite{Hammerstein_EAv_2010} directly, one can easily check that $(\frac{2}{a \sqrt{\A-\C \theta^2}})^{\lambda}\frac{2K_{\lambda}(a \sqrt{\A-\C\theta^2})}{\Gamma(-\lambda)}\rightarrow 1$ when $\theta^2\rightarrow \A/\C$. Therefore $Q(-\sqrt{\frac{\A}{\C}})=e^{-\sqrt{\A \C}}$. In this case, it is not clear if $q_{min}=-\sqrt{\frac{\A}{\C}}$ (the solution $x^{\star}$ is irregular) or $q_{min}\in (-\sqrt{\frac{\A}{\C}}, 0)$ (the solution $x^{\star}$ is regular).

Now let us look at the case $Z\sim GIG(\lambda, a, b)$ when $a>0, b>0$. Again from Proposition 9 of \cite{Hammerstein_EAv_2010} we have $\LL_Z(s)=(\frac{b}{\sqrt{b^2+2s}})^{\lambda}\frac{K_{\lambda}(a\sqrt{b^2+2s})}{K_{\lambda}(ab)}$ and $Q(\theta)=e^{\C \theta}(\frac{b}{\sqrt{b^2+\A-\C \theta^2}})^{\lambda}\frac{K_{\lambda}(a\sqrt{b^2+\A-\C \theta^2})}{K_{\lambda}(ab)}$. In this case $\s=-b^2/2$ and $\tha=\sqrt{\frac{\A+b^2}{\C}}$. One can easily check $\LL_Z(\s)=+\infty$ in this case. Therefore the unique optimal solution for (\ref{L22}) is given by (\ref{themainn}) and it is regular. 
\end{example}

\begin{corollary}\label{cor-stable} Consider the model (\ref{one}) with $\gamma=0$. In this case the distribution of $X$ is Elliptical distribution. Assume the CV-L of the mixing distribution $Z$ is $\s=0$. Then the corresponding optimization problem (\ref{L22}) has a unique solution $x^{\star}=0$. The CV-L of $Z$ is $\s=0$ if $EZ^n=+\infty$ for some positive integer $n$.
\end{corollary}
\begin{proof} Observe that in this case $\A=0$ and therefore $\tha=0$. Then $[-\tha, \tha]=\{0\}$. Therefore $q_{min}$ in (\ref{themain}) is $q_{min}=0$. As $\gamma=0$ also by assumption, we have $x^{\star}=0$ by (\ref{themain}). It is clear that this solution is unique. If $\s\neq 0$, then the Laplace transformation of $Z$ is finite in $(-\infty, |\s|)$ and this would imply that all the moments of $Z$ is finite. Therefore 
infinity of one of the  moments of $Z$ imply $\s=0$.
\end{proof}

\begin{corollary} Consider the model (\ref{one}) in high dimension, i.e., $d>0$. Assume $\mu-\1r_f$ and $\gamma$ are linearly independent.  Assume the mixing distribution $Z$ is strictly positive and $EZ<\infty$. Then the optimal solution of the problem (\ref{L22}) is explicitly given by
\[
x^{\star}=\frac{1}{aW_0}[\Sigma^{-1}\gamma +\frac{1}{EZ}\Sigma^{-1}(\mu-\1r_f)].
\]
\end{corollary}
\begin{proof} From \cite{takaki}, the optimal solution takes the following form 
\[
\bar{x}^{\star}=C[\Sigma^{-1}\gamma EZ+\Sigma^{-1}(\mu-\1 r_f)],
\]
for some constant $C>0$. We equate this solution with the solution (\ref{themain}), i.e., $\bar{x}^{\star}=x^{\star}$. The stated linear independence of $\mu-\1r_f$ and $\gamma$ then implies that
$q_{min}=-\frac{1}{EZ}$ in (\ref{themain}).
\end{proof}

\begin{example} (Stable distributions)\label{4.6} Lets look at the case of $\alpha-$stable distributions. Here we look at the  1- parametrization of the stable distributions (see Definition 1.5 of \cite{Nolan}). For other parametrizations see \cite{Nolan}. A distribution $W$ follows $\alpha-$stable distribution with parameters $\alpha\in (0, 2]$, $\beta\in [-1, 1]$, $\sigma>0$, $u\in \R$ and we write $W\sim S(\alpha, \beta, \sigma, u)$ if its characteristic function is given by 
\begin{equation}\label{cf-stable}
\phi(t)=Ee^{itW}=
\left \{  \begin{array}{ll} 
e^{-\sigma^{\alpha}|t|^{\alpha}\left [1-i\beta sign(t)tan(\frac{\pi \alpha}{2})\right ]+itu}      &  \mbox{$\alpha\neq 1$},\\
 e^{-\sigma |t|\left [1+i\beta \frac{2}{\pi}sign(t)\ln |t|\right ]+itu}     & \mbox{$\alpha=1$}.
 \end{array} 
\right. 
\end{equation}
When $\alpha=2$, a stable distribution is a Normal distribution.  When $\alpha \in (0, 2)$, $EW^2=+\infty$ for all $\beta \in [-1, 1], \sigma>0, u\in \R$. Therefore for the mixing distributions $Z=|W|, \alpha \in (0, 2), \beta \in [-1, 1], \sigma>0, u\in \R$, the corresponding CV-L is $\s=0$. Thus when $\gamma=0$ and when $Z=|W|, \alpha \in (0, 2), \beta \in [-1, 1], \sigma>0, u\in \R,$ in the model (\ref{one}),  the optimization problem (\ref{L22}) has a unique solution $x^{\star}=0$. This means that when the mixing distribution $Z$ in (\ref{one}) is equal to the absolute value of a stable distribution with $\alpha \in (0, 2)$ and when $\gamma=0$, then the optimal portfolio for an exponential utility maximizer is to invest all her/his wealth into the risk-free asset.
\end{example}

\begin{remark}
Stable distributions are infinitely divisible.
The characteristic functions (\ref{cf-stable}) of the stable laws can be obtained directly from their L\'evy-Khintchine representations. The generelized central limit theorem states that stable laws are the only non-trivial limits of normalized sums of independent
identically distributed random variables. As such they were proposed to model many empirical (heavy tails, skewness etc.) financial phenomenons in the past. The heavy tailedness of them is related with the CV-L of them being $\s=0$.  The above example \ref{4.6} shows that time-changed Brownian motion models with stable subordinators (the ones with Elliptical marginal distributions) always give the trivial portfolio, investing everything on the risk-free asset, as the optimal portfolio for an exponential utility maximizer. 
\end{remark}

As pointed out in Remark \ref{4.2}, our Lemma \ref{HH} shows that the solution for the problem (\ref{L22}) is unique. Part b) of this  Lemma shows that $\theta=0$ is not the minimizing point of the function $Q(\theta)$ under the condition that $\A \neq 0$ or $\s\neq 0$. For this unique  minimizing point  $\theta\neq 0$ of $Q(\theta)$ the first order condition (\ref{FOCQ}) can equivalently be written as
\begin{equation}
 \frac{\LL_Z'(\frac{\A}{2}-\frac{\C}{2}\theta^2)}{\LL_Z(\frac{\A}{2}-\frac{\C}{2}\theta^2)}=\frac{1}{\theta}.   
\end{equation}
A change of variable $\eta=\A/2-(\C/2)\theta^2$, which gives $\theta=-\sqrt{(\A-2\beta)/\C}$ due to 
$\theta<0$ by Lemma \ref{HH},  then gives
\begin{equation}\label{L-beta}
\frac{\LL_Z'(\beta)}{\LL_Z(\beta)}=-\sqrt{\C/(\A-2\beta)}, \; \; \s<\beta <\A/2.   
\end{equation}
From this we can conclude that if $x^{\star}$ is a regular solution for (\ref{L22}), then $\beta_{min}=:\A/2-(\C/2)q_{min}^2$ with $q_{min}$ in (\ref{themainn}) satisfies the relation (\ref{L-beta}). This observation is useful if it can be confirmed that the solution for the equation (\ref{L-beta}) is unique. Then this unique solution equals to $\beta_{min}$. Consider for example the case $Z=1$ in the model (\ref{one}). As discussed in Example  \ref{5.3ex} above, in this case we have $\LL_Z(s)=e^{-s}$. Then $\LL_Z'(\beta)/\LL_Z(\beta)=-1$ and it is clear that the equation $1=\sqrt{\C/(\A-2\beta)}$ has a unique solution $\beta=\A/2-\C/2$. This implies $q_{min}^2=1$ which then shows $q_{min}=-1$ is the minimizing point of $Q(\theta)$.

A positive random variable $Z$ is a GGC with generating pair $(\tau, \nu)$ if 
\begin{equation}\label{GGClap}
 \mathcal{L}_Z(s)=Ee^{-sZ}=e^{-\tau -\int_0^{\infty}\ln (1+\frac{s}{z})\nu(dz)}.  
\end{equation}
If $Z$ is a GGC with generating pair $(\tau, \nu)$, then $\frac{\LL_Z'(\beta)}{\LL_Z(\beta)}=-\tau-\int_0^{+\infty}\frac{1}{t-\beta}\nu(dt)$. So if the solution for (\ref{L22}) is regular,  then the $\beta_{min}$ defined above satisfy the following equation 
\[
-\tau-\int_{|\s|}^{+\infty}\frac{1}{t-\beta} \nu(dt)=-\sqrt{\C/(\A-2\beta )},
\]
where $\s$ is the CV-L of the GGC random variable $Z$.

Now consider the case of positive  $\alpha$-stable random variables $Z=S(\alpha, 1, \sigma, u), 0<\alpha < 1, u>0$. Here we took $\beta=1$ (see lemma 1.1 of \cite{Nolan}). After normalization these mixing distributions have the Laplace transformation $\LL_Z(s)=e^{-s^{\alpha}}$ (see Proposition 1 of \cite{Bondesson-ProbTheory} and also see \cite{Thomas-Simon}). Thus we have $\LL_Z'(s)/\LL_Z(s)=-s^{\alpha}\ln s$. Assume the problem (\ref{L22}) has regular solution (a necessary condition for this is $\gamma \neq 0$, see Corollary \ref{cor-stable}). Let $\beta_{min}=\A/2-(\C/2)q_{min}^2$ with $q_{min}$ in (\ref{themainn}). Then $0<\beta_{min}<\A/2$ and it satisfies the following equation due to (\ref{L-beta})
\[
\beta^{\alpha}\ln \beta=\sqrt{\C/(\A-2\beta)}.
\]
We square both sides of this equation and obtain  
\[
\A \beta^{2\alpha}(\ln \beta)^2-2\beta^{2\alpha+1}(\ln \beta)^2=\C.
\]
As discussed earlier, if this equation has a unique solution $\beta$ then it is $\beta_{min}$.

\begin{remark}
We should mention here that the formula (\ref{themain}) for the optimal portfolio for the problem (\ref{L22}) is related with the Laplace transformation of the mixing distribution $Z$ in the model (\ref{one}) only. Namely we don't need to know the probability density function of $Z$ to find the optimal portfolio for the optimization problem (\ref{L22}). The relation (\ref{L-beta}) gives a convenient approach to locate the unique optimal portfolio as discussed earlier.
\end{remark}

Next, we discuss the applications of our results in continuous time financial modelling. First we recall the Lemma 2.6 of \cite{Hammerstein_EAv_2010} here. According to this Lemma, for each model $F=N_d(\mu+\gamma z, z\Sigma)\circ G$ in (\ref{one}) there is a corresponding L\'evy process 
\begin{equation}\label{79}
 Y_t=\mu t+\gamma \tau_t+\bar{B}_{\tau_t},    
\end{equation}
with $Law(Y_1)=F$ and $Law(\tau_1)=G$ as long as $G\in \mathcal{J}$ (note that if $G\in \mathcal{J}$ then $X\in \mathcal{J}$ also from Lemma 2.5 of \cite{Hammerstein_EAv_2010}). In the model (\ref{79}), $(\bar{B}_t)_{t\geq 0}=(AB_t)_{t\geq 0}$ where $B_t$ is an $n-$dimensional  standard Brownian motion independent from $(\tau_t)_{t\geq 0}$ and $(\tau_t)_{t\geq 0}$ is a subordinator (a non-negative L\'evy process with increasing sample paths). We denote the L\'evy measure of this subordinator by $\rho$ and its Laplace transformation by
\begin{equation}\label{73}
 \mathcal{L}_{\tau_t}(s)=e^{-t\Psi(s)},   
\end{equation}
where $\Psi(s)=bs+\int_0^{\infty}(1-e^{-sy})\rho(dy)$ with a constant $b\geq 0$. As stated in Proposition 2.3 of \cite{Steven-Lalley}, the function $\Psi(s)$ is continuous, nondecreasing, nonnegative, and convex. At each time point $t>0$ we have
\begin{equation}\label{74}
 Y_t\overset{d}{=}\mu t+\gamma \tau_t+\sqrt{\tau_t}AN_d.  
\end{equation}

Now consider a  market with $n$ risky assets with price process $S_t\in \R^d$ and one risk-free asset with price process $B_t=e^{tr_f}$. Assume the log return process $Y_t=(Y_t^{(1)}, Y_t^{(2)}, \cdots, Y_t^{(d)})$, where $Y_t^{(i)}=\ln (S_t^{(i)}/S_0^{(i)})$ has the dynamics as in (\ref{79}). The log return in the risk-free asset is $\ln(B_t/B_0)=r_ft$. An exponential utility maximizer wants to determine the optimal portfolio at each time point $t$ based on the log return vector of risky assets $R\in \R^d$ with components $R^{(i)}=\ln (S_{t+\triangle}^{(i)}/S_t^{(i)})$ and the log-return of the  risk-free asset $R^{(0)}=\ln(B_{t+\triangle }/B_t)=\triangle r_f$ in the time horizon $[t, t+\triangle]$. Assume the time increment is $\triangle=1$. Then we have
\begin{equation}
 R\overset{d}{=}\mu +\gamma \tau_1+\sqrt{\tau_1}AN_d,   
\end{equation}
and from our Theorem \ref{2.5} the exponential utility maximizer's optimal portfolio at time $t$ is
\begin{equation}\label{76}
x_t^{\star}=\frac{1}{aW_0^{(t)}}\Big [\Sigma^{-1}\gamma -q_{min}^{(t)}\Sigma^{-1}(\mu-\1 r_f)\Big ],    
\end{equation}
where $W_0^{(t)}$ is his (initial) wealth that he invests on the $n+1$ assets for the period $[t, t+\triangle]$ and $q_{min}^{(t)}$ in (\ref{76}) is given by $q_{min}^{(t)}=argmin_{\theta\in \Theta}Q(\theta)$ in the corresponding domain $\theta$. Here 
\begin{equation}
 Q(\theta)=e^{C\theta-\Psi(\frac{1}{2}A-\frac{\theta^2}{2}C)}, 
\end{equation}
due to (\ref{73}).

\begin{example}\label{exam4.14} (Variance-gamma model) Consider the financial market that was discussed in the paper \cite{Madan_Dilip_B_And_Carr_Peter_P_And_Chang_Eric_C_1998}.  The stock price is given by $S(t)=S(0)e^{mt+X(t;\; \sigma_S, \; \nu_S, \; \theta_S)+\omega_St}$ in their equation (21), where $m$ is the mean-rate of return on the stock under the statistical probability measure, $\omega_S=\frac{1}{\nu_S}\ln (1-\theta_S\nu_S-\sigma_S^2\nu_S/2)$, and $X(t; \sigma_S, \nu_S, \theta_S)=b(\gamma(t; 1, \nu_S); \theta_S, \sigma_S)$
with $b(t; \theta, \sigma)=\theta t+\sigma W(t)$ being a Brownian motion with drift $\theta$ and volatility $\sigma$. Here the gamma process $\gamma(t; \mu, \nu)$ has mean rate $\mu$ and variance rate $\nu$ (note here that $\gamma(t; \mu, \nu)\sim G(\mu^2/\nu, \nu/\mu)$ with our notation for gamma random variables in this paper). The increment $g_0=:\gamma(t+1; 1, \nu_S)-\gamma(t; 1, \nu_S)\overset{d}{=}\gamma(1; 1, \nu_S)$ of this process has the Laplace transformation 
\begin{equation}
 \mathcal{L}_{g_0}(s)=(\frac{1}{1+s\nu_S})^{\frac{1}{\nu_S}},   
\end{equation}
which can be seen also from the characteristic function expression in (3) of \cite{Madan_Dilip_B_And_Carr_Peter_P_And_Chang_Eric_C_1998} for gamma processes. The risk-free asset in this financial market is given by $B_t=B_0e^{tr_f}$. The log returns of these two assets in the time horizon $[t, t+1]$ is given by
\begin{equation*}
\begin{split}
R=:&\ln (S(t+1)/S(t))\overset{d}{=}m+\omega_S+\theta_S\gamma(1; 1, \nu)+\sigma_S \sqrt{\gamma(1; 1, \nu_S)}N(0, 1),\\
 R^{0}=:&\ln (B_{t+1}/B_t)=r_f.
\end{split}
\end{equation*}

An exponential utility maximizer with utility function $u(x)=-e^{-ax}, a>0,$ and wealth $W_0^{(t)}$ at time $t$ wants to decide on the optimal proportion $x^{\star}$ on the risky asset of his wealth for the period $[t, t+1]$. His acceptable set for $x^{\star}$ is given by  
\begin{equation}\label{Exsa}
  S_a=\{x\in \R: aW_0^{(t)}\theta_S x-\frac{a^2(W_0^{(t)})^2}{2}\sigma^2_Sx^2>-\frac{1}{\nu_S}\}, 
\end{equation}
as $\s=-\frac{1}{\nu_S}$ in this case. 
The corresponding expressions for $\A, \B, \C$ in (\ref{ABC})
are given by 
\[
\A=(\frac{\theta_S}{\sigma_S})^2, \C=(\frac{m+\omega_S-r_f}{\sigma_S})^2, \B=\frac{\theta_S (m+\omega_S-r_f)}{\sigma^2_S}.
\]
Since the mixing distribution is a gamma random variable, the solution for the corresponding problem (\ref{L22}) is regular. Our Theorem \ref{2.5} shows that the optimal portfolio is given by
\begin{equation}\label{exstarn}
x^{\star}=\frac{1}{aW_0}[\frac{1}{\sigma^2_S}\theta_S-q_{min}\frac{1}{\sigma^2_S}(m+\omega_S-r_f)].
\end{equation}
where $q_{min}=argmin_{\theta \in (-\tha, \tha)} Q(\theta)$ with
$Q(\theta)$ given by (\ref{H}). Here $\tha=\sqrt{\frac{\A+2/\nu_S}{\C}}$. Next, we calculate $q_{min}$ explicitly. We have $Q(\theta)=e^{\C \theta}\LL_{g_0}(\A/2-(\C/2) \theta^2)$ and from this we get $\ln Q(\theta)=C\theta-\frac{1}{v_S}\ln (1+\frac{A}{2}v_S-\frac{C}{2}v_S\theta^2)$. The first order condition for the minimizing point of $\ln Q(\theta)$ gives $(\theta+\frac{1}{\C\nu_S})^2=\frac{1+\C\nu_S(2+\A \nu_S)}{\C^2\nu_S^2}$. This gives two solutions $\theta=-\frac{1}{\C\nu_S}\pm \frac{1}{\C\nu_S}\sqrt{1+\C\nu_S(2+\A \nu_S)}$. But since $\theta$ needs to be negative due to Lemma \ref{HH}, we take $q_{min}=\theta=-\frac{1}{\C\nu_S}-\frac{1}{\C\nu_S}\sqrt{1+\C\nu_S(2+\A \nu_S)}$. We then plug this into (\ref{themain}) and obtain
\begin{equation}
 x^{\star}=\frac{1}{aW_0^{(t)}\sigma^2_S}\Big [\theta_S+\frac{m+\omega_S-r_f}{\C\nu_S}+\frac{m+\omega_S-r_f}{\C\nu_S}\sqrt{1+\C\nu_S(2+\A\nu_S)}\Big ].
\end{equation}
Therefore in this case we have closed form expression for the optimal portfolio. We should mention that one can use similar calculations to obtain closed form expression for optimal portfolio in a market where risky assets are modelled by multi-dimensional variance gamma (MVG) model, see \cite{Madan_Dilip_B_And_Seneta_Eugene_1990} for the details of MVG models.
\end{example}

\begin{remark} Price processes with log-returns of the type (\ref{79}) has been quite popular in financial literature in the past. Such models include inverse Gaussian L\'evy processes, hyperbomic L\'evy motions, variance gamma models, and CGYM models and all of these models were shown to fit empirical data quite well, see \cite{CGMY, Geman, Schoutens_Wim_2003, Eberlein_Ernst_And_Keller_Ulrich_1995, Madan-Yor-Timechange} and the references therein for this. In fact, every semimartingale can be written as a time change of Brownian motion, see \cite{Monroe} for this. This means that all the L\'evy processes are time change of Brownian motion. In all these cases, if the time changing subordinator is independent from the Brownian motion then  our Theorem \ref{2.5} is applicable in principle. However, it is not easy to find the time-change used for general semimartinagles. Recently the paper \cite{Madan-Yor-Timechange} obtained the time change used for the CGMY model and   Meixner processes. Our results in this paper can be applied to such processes  to determine optimal portfolios for an exponential utility maximizer in a market where single or multiple risky asset dynamics follow such models. 
\end{remark}

\section{Conclusion}
The main result of this paper is Theorem \ref{2.5} where we show that the problem of locating the optimal portfolio for (\ref{L2}) when the utility function is exponential boils down to finding the minimum point of a real valued function on the real-line, improving the Theorem 1 of \cite{Birge-Bedoya} for the case of GH models and in the mean time extending it from the class of GH models to the general class of NMVM models. Our Theorem \ref{mainu} shows that optimal exponential utility in small markets converge to the overall best exponential utility in the large financial market. While optimal portfolio problems under expected utility criteria for exponential utility functions have been discussed extensively in the past financial literature, an explicit solution of the optimal portfolio as in Theorem \ref{2.5} above seems to be new. 
This is partly due to the condition we impose on the return vector $X$ of being a NMVM model. However, despite this restrictive condition on $X$, asset price dynamics with NMVM distributions in their log returns often show up in financial literature like  exponential variance gamma and exponential generalized hyperbolic L\'evy motions.

\bibliographystyle{plainnat}

\bibliography{main}

\vspace{0.2in}

\end{document}